\documentclass[11pt]{article}
\usepackage[american]{babel}
\usepackage{a4}
\usepackage{amsmath, amssymb, amsthm}

\theoremstyle{plain}

\newtheorem{theorem}{Theorem}
\newtheorem{proposition}{Proposition}

\theoremstyle{definition}

\newtheorem*{definition}{Definition}
\newtheorem*{remark}{Remark}

\usepackage{hyperref}
\usepackage{subfig}
\usepackage{ifthen}
\usepackage{tikz}
\usetikzlibrary{decorations.markings,positioning,arrows,patterns,shapes}
\usepackage{dsfont}

\DeclareMathOperator{\re}{Re}
\DeclareMathOperator{\im}{Im}
\DeclareMathOperator{\sn}{sn}
\DeclareMathOperator{\artanh}{artanh}
\newcommand{\Li}[1][2]{\operatorname{Li}_{#1}} 

\newcommand{\cG}{{\mathcal G}}
\newcommand{\mC}{{\mathds C}}
\newcommand{\mR}{{\mathds R}}
\newcommand{\mZ}{{\mathds Z}}

\begin{document}

\title{On discrete integrable equations with convex variational principles}

\author{Alexander I.~Bobenko\footnote{Institut f\"ur Mathematik, MA 8-3, Technische
Universit\"at Berlin, Stra{\ss}e des 17. Juni 136, 10623 Berlin, Germany.}$\;^{,1}$ \and Felix G\"unther\footnotemark[1]$\;^{,2}$}

\date{}
\maketitle

\footnotetext[1]{Supported by the DFG Research Unit ``Polyhedral Surfaces'' and the SFB/TR 109 ``Discretization in geometry and dynamics''. E-mail: bobenko@math.tu-berlin.de}
\footnotetext[2]{Supported by the DFG Research Unit ``Polyhedral Surfaces'' and the Deutsche Telekom Stiftung. E-mail: fguenth@math.tu-berlin.de}

\begin{abstract}
\noindent
We investigate the variational structure of discrete Laplace-type equations that are motivated by discrete integrable quad-equations. In particular, we explain why the reality conditions we consider should be all that are reasonable, and we derive sufficient conditions (that are often necessary) on the labeling of the edges under which the corresponding generalized discrete action functional is convex. Convexity is an essential tool to discuss existence and uniqueness of solutions to Dirichlet boundary value problems. Furthermore, we study which combinatorial data allow convex action functionals of discrete Laplace-type equations that are actually induced by discrete integrable quad-equations, and we present how the equations and functionals corresponding to $(\textnormal{Q}3)$ are related to circle patterns.\\ \vspace{0.5ex}

\vspace{0.5ex}

\end{abstract}

\raggedbottom
\setlength{\parindent}{0pt}
\setlength{\parskip}{1ex}


\section{Introduction}\label{sec:intro}

Discrete integrable systems on quad-graphs serve as discretizations of integrable partial differential equations with a two-dimensional space-time, as suggested by the first author and Suris in \cite{BS02}. They identified integrability of such systems with their multi-dimensional consistency, as did also Nijhoff in \cite{N02}. This property was used by the first author, Adler, and Suris in \cite{ABS03} to classify integrable systems on quad-graphs, resulting in the now famous ABS list of discrete integrable equations.

Multi-dimensional consistent quad-equations can be consistently imposed on all two-dimensional sublattices of $\mZ^d$; i.e., the corresponding multidimensional system possesses solutions such that its restrictions to two-dimensional quad-surfaces are solutions of the original two-dimensional system. In \cite{ABS03}, the first author, Adler, and Suris described the variational (Lagrangian) formulation of such systems. Being more precise, they showed that solutions of some discrete integrable quad-equations of the ABS list on a quad-surface $\Sigma$ are critical points of a certain discrete action functional $S=\int_\Sigma \mathcal{L}$, where $\mathcal{L}$ is a suitable discrete two-form on $\Sigma$, called the Lagrangian. Lobb and Nijhoff found out in \cite{LN09} that the discrete action functional takes the same value if the underlying surface is changed only locally, meaning that $\mathcal{L}$ is closed. The first author and Suris extended these results to all equations from the ABS list and gave a more conceptual proof in \cite{BS10}. Variational principles are a very powerful tool for numerical simulations of problems in classical mechanics. For a presentation of discrete mechanics along with a discrete variational principle and numerical applications, we refer to the paper \cite{MaW97} of Marsden and Wendlandt.

Following the papers \cite{BoSu14} of the first author and Suris, and \cite{BPSu14} of Boll, Petrera, and Suris, the idea of Lobb and Nijhoff can be equivalently stated as follows: The solutions of discrete integrable systems give critical points simultaneously for discrete action functionals along all possible two-dimensional quad-surfaces of $\mZ^d$, and the Lagrangian is closed on solutions. This interpretation corresponds to the classical theory of pluriharmonic functions. A pluriharmonic function is defined as a real-valued function of several complex variables that minimizes the Dirichlet energy along any holomorphic curve in its domain. The relation to pluriharmonic functions motivated the first author and Suris to introduce pluri-Lagrangian problems: Given a $k$-dimensional Lagrangian $\mathcal{L}$ in $d$-dimensional space that depends on a sought-after function $u$ of $k$ variables, one looks for functions $u$ that deliver critical points to $S=\int_\Sigma \mathcal{L}$ on a $k$-dimensional surface $\Sigma$. They claimed that integrability of variational systems should be understood as the existence of the pluri-Lagrangian structure.

A general theory of one-dimensional pluri-Lagrangian systems was developed by Suris in \cite{Su13}; two-dimensional systems were discussed by Boll, Petrera, and Suris in \cite{BPSu14}. They identified the main building blocks of the discrete Euler-Lagrange equations of the pluri-Lagrangian systems as the corresponding equations at the quad-surface consisting of three elementary squares adjacent to a vertex of valence three, so-called 3D-corner equations. They discussed the notion of consistency of the system of 3D-corner equations and analyzed it for a special class of three-point two-forms motivated by discrete integrable quad-equations of the ABS list.

Note that quantization of discrete integrable systems on quad-graphs yields solvable lattice models. Here, the consistency principle corresponds to the quantum Yang-Baxter equation. Classical discrete integrable systems on quad-graphs are then recovered in the quasi-classical limit. The corresponding action functional is derived as a quasi-classical limit of the partition function of the corresponding integrable quantum model (the Lagrangians being the quasi-classical limit of the Boltzmann weights). A quantization of circle patterns, with the corresponding quantization of the action functional introduced by the first author and Springborn in \cite{BSp04}, was carried out by Bazhanov, Mangazeev, and Sergeev in \cite{BaMSe07}. The quantum ``master solution'' that serves as a quantization of the discrete Laplace-type equation corresponding to the $(\textnormal{Q}4)$-system was recently introduced by Bazhanov and Sergeev in \cite{BaSe10,BaSe11}. Functionals obtained in the quasi-classical limit are related to the variational problems considered in this paper.

Our setting introduced in Section~\ref{sec:description} is a finite bipartite quad-graph embedded in an oriented surface, on which we consider discrete Laplace-type equations that are motivated from discrete integrable quad-equations of the ABS list. We explain why the reality conditions, i.e., restrictions of the variables and parameters to certain one-dimensional complex subspaces, should be actually all that are reasonable, and we compute the corresponding discrete action functionals. Our variational description turns out to be slightly more general than the previous formulation of the first author and Suris in \cite{BS10}, because we include boundary terms yielding a nonintegrable generalization.

The generalized discrete action functionals are studied in Section~\ref{sec:action}. We are particularly interested in conditions on the parameters of the quad-equations under which the corresponding generalized discrete action functional is strictly convex (or strictly concave). For the generalized discrete action functionals corresponding to any integrable quad-equation except $(\textnormal{Q}4)$, we give necessary and sufficient conditions for strict convexity; in the case of $(\textnormal{Q}4)$, we find at least sufficient conditions. These results are summarized in Section~\ref{sec:main}.

Strict convexity implies the uniqueness of solutions to Dirichlet boundary value problems and helps to investigate their existence in Section~\ref{sec:Q_Dirichlet}. Moreover, minimization of the corresponding functional is an effective tool to construct the corresponding solution numerically. For example, Stephenson's program \verb|circlepack| constructs circle packings by a method of Thurston that can be interpreted as a particular method to minimize the action functional of circle patterns that was derived by the first author and Springborn in \cite{BSp04}. By relating the discrete generalized action functionals corresponding to discrete Laplace-type equations of type $(\textnormal{Q}3)$ to the circle pattern functionals of the first author and Springborn, we can adapt their existence results to our setting. In addition, we prove that the Dirichlet boundary value problems in the case of $(\textnormal{Q}4)$ with rectangular or rhombic lattices can be uniquely solved.

Having identified situations under which the generalized discrete action functional is strictly convex, we investigate in Section~\ref{sec:integrable_cases} when discrete integrable equations may lead to convex variational principles. We discuss the influence of combinatorial data, but our presentation is far from being complete.

Finally, we present a geometric interpretation of the solutions to the (generalized) discrete Laplace-type equations corresponding to $(\textnormal{Q}3)$ as Euclidean and hyperbolic circle patterns with conical singularities in Section~\ref{sec:circle_patterns}. Note that the convexity of the functionals plays a crucial role in the variational description of circle patterns by the first author and Springborn in \cite{BSp04}. Our geometric interpretation of $(\textnormal{Q}3)$ turns out to be closely connected to the paper \cite{BoMeSu05} of the first author, Mercat, and Suris. In their paper, they discussed integrable circle patterns and discrete integrability of a system of cross-ratio equations. Their notions of integrability fit well to ours in terms of the parameters of the discrete Laplace-type equations.


\section{Discrete Laplace-type integrable equations}\label{sec:description}

In Section~\ref{sec:quadeq}, we first introduce basic definitions and notations we will use in the sequel and comment on the discrete integrable quad-equations we are investigating. How these equations induce discrete Laplace-type equations is the topic of Section~\ref{sec:Laplace}, as well as their variational description. In the end, we derive the reality conditions discussed later in Section~\ref{sec:action} and explain why they are most likely the only reasonable conditions to that methods of real variational calculus can be applied.


\subsection{Discrete integrable quad-equations}\label{sec:quadeq}

We consider a finite connected strongly regular bipartite \textit{quad-graph} $\cG$ embedded in an oriented surface (mainly the plane) such that $0$-cells correspond to vertices $V(\cG)$, $1$-cells to edges $E(\cG)$, and $2$-cells to quadrilateral faces $F(\cG)$. We refer to the maximal independent sets of vertices of $\cG$ as \textit{white} and \textit{black} vertices. The assumption of strong regularity asserts that the boundary of a quadrilateral contains a particular vertex or a particular edge at most once and that two different edges or faces have at most one vertex or edge in common, respectively. As a consequence, if a line segment connecting two vertices of $\cG$ is the (possibly outer) diagonal of a quadrilateral, there is just one such face of $\cG$. Let $\cG_W$ and $\cG_B$ be the graphs defined on the white and black vertices where the edges are exactly corresponding diagonals of faces of $\cG$. We will always assume that $\cG_W$ is connected and contains at least two vertices.

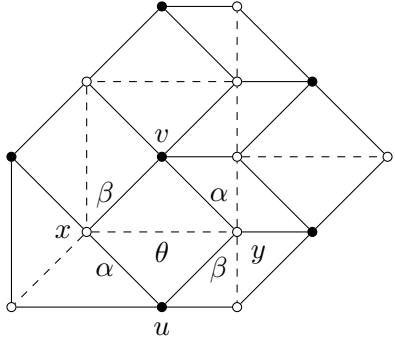
\begin{figure}[htbp]
\begin{center}
\beginpgfgraphicnamed{quad}
\begin{tikzpicture}
[white/.style={circle,draw=black,fill=white,thin,inner sep=0pt,minimum size=1.2mm},
black/.style={circle,draw=black,fill=black,thin,inner sep=0pt,minimum size=1.2mm}]
\node[white] (w1) [label=left:$x$]
at (-1,-1) {};
\node[white] (w2) [label=below right:$y$]
 at (1,-1) {};
\node[white] (w3)
 at (1,0) {};
\node[white] (w4)
 at (1,1) {};
\node[white] (w5)
 at (-1,1) {};
\node[white] (w6)
 at (-2,-2) {};
\node[white] (w7)
 at (1,-2) {};
\node[white] (w8)
 at (3,0) {};
\node[white] (w9)
 at (1,2) {};
\node[black] (b1) [label=above:$v$]
 at (0,0) {};
\node[black] (b2)
 at (2,1) {};
\node[black] (b3)
 at (2,-1) {};
\node[black] (b4) [label=below:$u$]
 at (0,-2) {};
\node[black] (b5)
 at (-2,0) {};
\node[black] (b6)
 at (0,2) {};
\draw[dashed] (w1) --node[midway,below] {$\theta$} (w2) --  (w3) --  (w4) --  (w5) --  (w1);
\draw[dashed] (w1) --  (w6);
\draw[dashed] (w2) --  (w7);
\draw[dashed] (w3) --  (w8);
\draw[dashed] (w4) --  (w9);
\draw (b2) -- (w9) -- (b6) -- (w5) -- (b5) -- (w6) -- (b4) -- (w7) -- (b3) -- (w8) -- (b2);
\draw (b5) -- (w1) --node[midway,left] {$\alpha$}  (b4) --node[midway,right] {$\beta$}  (w2) -- (b3) -- (w3) -- (b2) -- (w4) -- (b6);
\draw (b1) --node[midway,left] {$\beta$}  (w1);
\draw (b1) --node[midway,right] {$\alpha$}  (w2);
\draw (b1) -- (w3);
\draw (b1) -- (w4);
\draw (b1) -- (w5);
\end{tikzpicture}
\endpgfgraphicnamed
\caption{Bipartite quad-graph with white graph (dashed lines)}
\label{fig:white_subgraph}
\end{center}
\end{figure}

On the set $E(\cG)$ of edges of $\cG$, we will investigate labelings with complex numbers such that opposite edges of a quadrilateral receive the same label. These labels will be the parameters of the quad-equations described below. For short, only such labelings are actually said to be \textit{labelings} of $E(\cG)$. In the notation of Figure~\ref{fig:white_subgraph}, any such labeling induces a labeling of $E(\cG_W)$ by $\theta:=\alpha-\beta$ that we say to be an \textit{induced labeling}. Conversely, not any labeling of $E(\cG_W)$ is induced by some labeling of $E(\cG)$.

\begin{definition}\label{def:labeling}
A labeling of $E(\cG_W)$ is said to be \textit{integrable} if it is induced by some labeling of $E(\cG)$.
\end{definition}

Given a quad-graph $\cG$ with a labeling of edges, we consider a collection of equations on each face of $\cG$ of the type
\begin{equation}\label{eq:Q}
 Q(x,u,y,v;\alpha,\beta)= 0,
\end{equation}
where $\alpha$ and $\beta$ are associated to the edges of the quadrilateral as in Figure~\ref{fig:white_subgraph} and $x,u,y,v$ are complex variables assigned to the four vertices. Note that $x$ has the meaning of a vertex as well as of the field variable associated to it, but the meaning will be clear from the context.

Here, $Q$ shall be affine linear in each variable and Equation~(\ref{eq:Q}) is assumed to be invariant under the symmetry group of the square. Integrability of this quad-equation is identified by its three-dimensional consistency, see Figure~\ref{fig:3D}. This means that the equation can be consistently imposed on the six faces of the cube, i.e., values $x,x_1,x_2,x_3$ uniquely define $x_{123}$. By assumption, $x,x_1,x_2,x_3$ uniquely determine $x_{12},x_{23},x_{13}$, but the three faces incident to $x_{123}$ yield three a priori different values for $x_{123}$. Three-dimensional consistency assures that these three values agree.

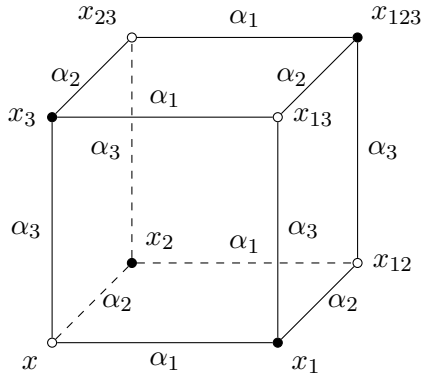
\begin{figure}[htbp]
\begin{center}
\beginpgfgraphicnamed{3D}
\begin{tikzpicture}
[black/.style={circle,draw=black,fill=white,thin,inner sep=0pt,minimum size=1.2mm},
white/.style={circle,draw=black,fill=black,thin,inner sep=0pt,minimum size=1.2mm},scale=1.5]
\node[black] (x) [label=below left:$x$] at (0,0) {};
\node[white] (x1) [label=below right:$x_1$] at (2,0) {};
\node[black] (x13) [label=right:$x_{13}$] at (2,2) {};
\node[white] (x3) [label=left:$x_3$] at (0,2) {};
\node[white] (x2) [label=above right:$x_2$] at (0.707,0.707) {};
\node[black] (x12) [label=right:$x_{12}$] at (2.707,0.707) {};
\node[white] (x123) [label=above right:$x_{123}$] at ((2.707,2.707) {};
\node[black] (x23) [label=above left:$x_{23}$] at (0.707,2.707) {};
\draw (x) -- node[midway,below] {$\alpha_1$} (x1) -- node[midway,right] {$\alpha_3$} (x13)-- node[midway,above] {$\alpha_1$} (x3)-- node[midway,left] {$\alpha_3$} (x);
\draw (x1) -- node[midway,right] {$\alpha_2$} (x12) -- node[midway,right] {$\alpha_3$} (x123)-- node[midway,left] {$\alpha_2$} (x13);
\draw (x123) -- node[midway,above] {$\alpha_1$} (x23) -- node[midway,left] {$\alpha_2$} (x3);
\draw[dashed] (x) -- node[midway,right] {$\alpha_2$} (x2) -- node[midway,above] {$\alpha_1$} (x12);
\draw[dashed] (x2) -- node[midway,left] {$\alpha_3$} (x23);
\end{tikzpicture}
\endpgfgraphicnamed
\caption{Three-dimensional consistency}
\label{fig:3D}
\end{center}
\end{figure}

In \cite{ABS03}, the first author, Adler, and Suris classified all these integrable equations under one additional assumption, called the tetrahedron property. This property demands that $x_{123}$ depends on $x_1,x_2,x_3$, but not on $x$. They have showed that up to M\"obius transformations of the field variables any such equation is equivalent to one of nine canonical equations. These nine equations are listed in the so-called ABS list. The ABS list consists of a list Q containing four equations, a list H of three, and a list A of two equations. As a result of their subsequent paper \cite{ABS09}, the tetrahedron property can be replaced by certain nondegeneracy assumptions, just giving the list Q.

In our setting, we assume that Equation~(\ref{eq:Q}) is an equation of the ABS list. It turns out that for the investigation of discrete Laplace-type equations, one can restrict to the nondegenerate equations of the list Q.

\begin{remark}
Let us shortly comment how our definition of integrability is motivated and how it is related to other uses. First, the way how a labeling of $E(\cG)$ induces a labeling of $E(\cG_W)$ is kind of a discrete derivative. In this sense, integrable labelings of $E(\cG_W)$ can be integrated to labelings of $E(\cG)$. But we are mainly interested in the connection to discrete integrable quad-equations. In the light of discrete Laplace-type equations and generalized discrete action functionals that we introduce in Section~\ref{sec:Laplace}, any labeling of $E(\cG_W)$ defines a discrete Laplace-type equation, but only an integrable labeling gives a discrete Laplace-type equation that actually comes from a discrete integrable quad-equation.

Note that integrability of $Q$, i.e., three-dimensional consistency, is related to admitting a discrete zero curvature representation, as was discussed by the first author and Suris in \cite{BS02} and by Nijhoff in \cite{N02}. In the case of circle patterns in Section~\ref{sec:circle_patterns}, our definition of integrability was denoted by \textit{integrability of the corresponding cross-ratio system} in the first author's and Suris' book \cite{BS08}. On the other hand, their definition of integrable circle patterns corresponds to isoradial realizations of them.
\end{remark}


\subsection{Discrete Laplace-type integrable equations and their variational structure}\label{sec:Laplace}

The first author, Adler, and Suris showed that each equation in the ABS list possesses a three-leg form \cite{ABS03}, which is additive for $(\textnormal{Q}1)_{\delta=0}$, $(\textnormal{H}1)$, and $(\textnormal{A}1)_{\delta=0}$ and multiplicative for the other equations. Here, an \textit{additive three-leg form} of Equation~(\ref{eq:Q}) (centered at $x$) is an equivalent equation
\begin{equation*}
 i\psi(x,u;\alpha) - i\psi(x,v;\beta)= i\varphi(x,y;\alpha-\beta)=i\varphi(x,y;\theta)
\end{equation*}
for some functions $\psi$ and $\varphi$ called (additive) \textit{three-leg functions}. If Equation~(\ref{eq:Q}) is equivalent to
\begin{equation*}
 \frac{\Psi(x,u;\alpha)}{\Psi(x,v;\beta)}= \Phi(x,y;\alpha,\beta)=\Phi(x,y;\theta)
\end{equation*}
for some functions $\Psi$ and $\Phi$, Equation~(\ref{eq:Q}) is said to possess a \textit{multiplicative three-leg form}, and $\Psi$ and $\Phi$ are called (multiplicative) \textit{three-leg functions}. Clearly, $i\psi=\log \Psi$, $i\varphi=\log \Phi$ induce an additive three-leg form modulo $2\pi$.

\begin{definition}
The functions $\varphi$ and $\Phi$ appearing in the additive or multiplicative three-leg form, respectively, are called \textit{long leg functions}; $\psi$ and $\Psi$ are called \textit{short leg functions}.
\end{definition}

\begin{remark}
Note that in most formulations of the three-leg forms, the factor of $i$ does not enter. The reason why we add it here is that we want to get real functionals in Section~\ref{sec:action}. For the reality conditions discussed in Section~\ref{sec:reality}, $i\varphi$ turns out to be (generally) a real function, leading to a real generalized discrete action functional later. Not taking the factor $i$ into account would necessitate considering the imaginary part of the functionals.
\end{remark}

For the discrete Laplace-type equations we define below, only the long leg functions are important. Since each equation from the lists A and H shares its long leg function with some equation from the list Q, we may restrict our attention to the latter. Note that $\psi=\varphi$ and $\Psi=\Phi$ for all equations from the list Q. In the appendix, we provide a list of discrete integrable quad-equations of the list Q together with their three-leg functions. There, we slightly adapt the formulation of the quad-equations to our setting, but mainly follow the presentation of the first author and Suris in the appendix of their paper \cite{BS10}. For the quad-equation $(\textnormal{Q}3)_{\delta=0}$, we even give two three-leg functions, leading to equivalent three-leg forms. The reason for that is that under one reality condition, we see a remarkable analogy to Euclidean circle patterns if we not take the three-leg function as it comes up in the work \cite{BS10} of the first author and Suris. 

In the original paper \cite{ABS03} of the first author, Adler, and Suris, the notation of equation $(\textnormal{Q}4)$ used Weierstra{\ss} elliptic functions. In the end of Section~\ref{sec:Q4}, we will use this formulation to investigate rhombic period lattices. But the Jacobian formulation discovered by Hietarinta in \cite{Hi05} seems to be more appropriate for rectangular period lattices and will be used in this chapter. In particular, the limit of $(\textnormal{Q}4)$ to $(\textnormal{Q}3)_{\delta=1}$ discussed in Section~\ref{sec:limit} becomes clearer.

The restriction of any solution of the system of discrete integrable equations on the quad-graph $\cG$ to the white graph $\cG_W$ satisfies for each inner white vertex $x$ the \textit{additive discrete Laplace equation}
\begin{equation}\label{eq:Laplace_Q10}
 \sum\limits_{e=(x,y_k)\in E(\cG_W)} \varphi(x,y_k;\theta_e)= 0
\end{equation}
in the case of $(\textnormal{Q}1)_{\delta=0}$ and the \textit{multiplicative discrete Laplace equation}
\begin{equation}\label{eq:Laplace_mult}
 \prod\limits_{e=(x,y_k)\in E(\cG_W)} \Phi(x,y_k;\theta_e)= 1
 \end{equation}
in the case of the other Q-equations. Here, $\theta_e$ denotes the label induced on the white edge $e$.

Equation~(\ref{eq:Laplace_mult}) corresponds to the equivalent form
\begin{equation}\label{eq:Laplace_add}
 \sum\limits_{e=(x,y_k)\in E(\cG_W)} \varphi(x,y_k;\theta_e)\equiv 0 \mod 2\pi,
 \end{equation}
where $i\varphi=i\varphi(x,y;\theta)$ is some branch of $\log \Phi$. Now, Equation~(\ref{eq:Laplace_add}) is satisfied if and only if there exists a real number $\Theta_x \equiv 0 \mod 2\pi$ such that
\begin{equation}\label{eq:Laplace}
 \sum\limits_{e=(x,y_k)\in E(\cG_W)} \varphi(x,y_k;\theta_e)=\Theta_x
 \end{equation}
holds. The additive discrete Laplace Equation~(\ref{eq:Laplace_Q10}) is included in Equation~(\ref{eq:Laplace_add}) by taking $\Theta_x=0$.

Even in the case that $x$ is not an inner white vertex of $\cG$, we have a corresponding discrete Laplace equation~(\ref{eq:Laplace}), but the restriction $\Theta_x \equiv 0 \mod 2\pi$ does not have to be satisfied and $\Theta_x$ could be nonreal. However, the reality conditions in Section~\ref{sec:reality} imply $\im \varphi \equiv 0$, so $\Theta_x \in \mR$.

\begin{definition}
For a general labeling of $E(\cG_W)$ and a general choice of $\Theta$, Equation~(\ref{eq:Laplace}) is said to be a \textit{discrete Laplace-type equation}. If the labeling is integrable and $\Theta_x \equiv 0 \mod 2\pi$ (in the case of $(\textnormal{Q}1)_{\delta=0}$, $\Theta_x=0$) for all inner vertices $x \in V(\cG_W)$, Equation~(\ref{eq:Laplace}) is said to be a \textit{discrete Laplace-type integrable equation}.
\end{definition}

We name the discrete Laplace-type equations that are not induced by discrete integrable system still by their corresponding integrable cases $(\textnormal{Q}1)$, $(\textnormal{Q}2)$, $(\textnormal{Q}3)$, and $(\textnormal{Q}4)$. These generalized equations seem to be important. For example, the generalized discrete action functional of $(\textnormal{Q}3)$ describes general circle patterns with conical singularities, as we will see in Section~\ref{sec:circle_patterns}. However, the integrable case is nevertheless important and is investigated in Section~\ref{sec:integrable_cases}.

According to the first author, Adler, and Suris \cite{ABS03}, there exists for any equation from the ABS list a change of variables $x=f(X)$, $u=f(U)$, $y=f(Y)$, $v=f(V)$ such that $\varphi$ possesses a symmetric primitive $L=L(X,Y;\theta)$ in the new variables, i.e.,
\begin{equation}\label{eq:symmetry}
 L(X,Y;\theta)=L(Y,X;\theta),
 \end{equation}
\begin{equation}\label{eq:derivative}
 \frac{\partial}{\partial X} L(X,Y;\theta)=\varphi(x,y;\theta).
\end{equation}
Since any two branches of the logarithm differ by a multiple of $2\pi i$ only, $L$ exists for any choice $i\varphi$ of the branch of $\log \Phi$. For convenience, we will consider $\varphi$ as a function of $X,Y,\theta$ in the following, using the notation $\varphi(X,Y;\theta)$ instead of $\varphi(x,y;\theta)=\varphi(f(X),f(Y);\theta)$.

Using the symmetry of $L$, the discrete Laplace-type Equations~(\ref{eq:Laplace}) for white vertices $x$ are the discrete Euler-Lagrange equations for the \textit{generalized discrete action functional}
\begin{equation}\label{eq:action}
 S=\sum\limits_{e=(x,y)\in E(\cG_W)} L(X,Y;\theta_e) - \sum\limits_{x\in V(\cG_W)} \Theta_x X.
 \end{equation}
Thus, critical points of (\ref{eq:action}) correspond exactly to solutions of (\ref{eq:Laplace}). Note that other choices of $\varphi$ fulfilling $\exp(i\varphi)=\Phi$ lead to analogous results, but $\Theta$ may change. Hence, we can always restrict to some branch of $\log \Phi$ of our choice.

\begin{remark}
The presentations of the first author, Adler, and Suris in \cite{ABS03} and of the first author and Suris in \cite{BS10} did not include the extra terms $\Theta$. But these terms become important in our consideration, when we will restrict to real variables and extend $\varphi$ smoothly. Since the manifold of solutions of Equation~(\ref{eq:Q}) is not connected in real space, $\varphi$ can never be part of a global additive three-leg form unless we consider $(\textnormal{Q}1)_{\delta=0}$. In contrast, the manifold of solutions of Equation~(\ref{eq:Q}) is connected in projective complex space $(\mathds{CP}^1)^4$ and an additive three-leg form exists globally.
\end{remark}

The main class of examples for discrete pluri-Lagrangian problems in dimension two considered in the paper \cite{BPSu14} of Boll, Petrera, and Suris is given by a Lagrangian $\mathcal{L}$ of the form \[\mathcal{L}(x,u,y,v;\alpha,\beta)=L'(X,U;\alpha)-L'(X,V;\beta)-L(X,Y;\alpha,\beta)\] with some functions $L,L'$. The Langrangian structure of discrete integrable quad-equations discussed by Lobb and Nijhoff \cite{LN09} as well as by the first author and Suris \cite{BS10} essentially corresponds to such an $\mathcal{L}$, where $L$ is the primitive of $\varphi$ introduced above ($\theta=\alpha-\beta$) and $L'$ is a primitive of the short leg function $\psi$. To ensure the skew-symmetry of the Lagrangian, \[\mathcal{L}(X,U,Y,V;\alpha,\beta)=\mathcal{L}(X,V,Y,U;\beta,\alpha),\] $L(X,Y;\theta)=-L(Y,X;-\theta)$ has to hold. In our considerations below, $L(X,Y;\theta)=-L(X,Y;-\theta)$ is generally true. But in addition, $L(X,Y;\pi-\theta)=-L(X,Y;\pi+\theta)$ holds in some instances of $(\textnormal{Q}3)$. This corresponds to the skew-symmetry $L(X,Y;\theta^*)=-L(X,Y;(-\theta)^*)$ with $\theta^*=\pi-\theta$ that was used by Bazhanov, Mangazeev, and Sergeev in \cite{BaMSe07}. 


\subsection{Reality conditions allowing variational principles}\label{sec:reality}

In this section, we will derive \textit{reality conditions}, i.e., restrictions of the variables $X$, $Y$ and the labels on edges of $\cG$ to certain one-dimensional connected submanifolds of $\mC$, under which the generalized discrete action functional takes values on an one-dimensional affine real subspace of $\mC$. Only in such cases we can ask for convexity of $S$ and can apply the methods of real variational calculus. Furthermore, we will explain why we believe that there are no further reasonable reality conditions than the ones considered in Section~\ref{sec:action}.

If $S$ takes its values on a line $R$ in $\mC$, $L(\cdot,\cdot;\theta)$ and thus its derivative $\varphi(\cdot,\cdot;\theta)$ take their values on lines parallel to $R$. Since the functions $\varphi$ (or $\varphi=-i\log \Phi$) that are given in the appendix are smooth functions with a discrete set of singularities, we can locally invert $\varphi$ with respect to $\theta$. So for fixed generic $X$ and $Y$, $\theta$ lies on some one-dimensional submanifold of $\mC$, at least locally. Using that $\theta=\alpha-\beta$ for an induced labeling, the set of all $\alpha-\beta$ is locally a one-dimensional submanifold. But then, the domains of $\alpha$ and $\beta$ have to be lines parallel to each other. It is natural to demand that all labels have the same domain of definition, implying that the induced labels $\theta$ lie on a line through the origin.

In the cases of $(\textnormal{Q}1)$ and $(\textnormal{Q}3)_{\delta=0}$, the three-leg functions depend just on $X-Y$ and $\theta$. Applying the same reasoning as above, $X$ and $Y$ lie on lines parallel to each other.

To determine possible reality conditions for the discrete Laplace-type equation corresponding to $(\textnormal{Q}1)_{\delta=0}$ is now quite simple. Since $X-Y$ lies on a line, $1/(X-Y)$ describes a circle unless $X-Y$ passes through the origin. But $\varphi$ has to lie on a line, so all reasonable reality conditions are $X,Y\in \mR \omega +c$, $\theta \in \mR \omega'$ for some complex numbers $\omega,\omega'\neq 0$ and $c$. By rotating all three-leg functions appropriately, we may assume $\omega=\omega'=1$. Furthermore, a variable transform $X \mapsto X-c$ for all $X$ does not change $\varphi$. Thus, it suffices to consider $X,Y,\theta \in \mR$ in Section~\ref{sec:Q10}.

For the other quad-equations of the list Q, we have to investigate the multiplicative three-leg functions $\Phi$. Since $\Phi=\exp(i\varphi)$ and the values of $\varphi$ lie on a line, the possible values of $\Phi$ lie either on the real line, a circle centered at the origin, or a spiral.

Let us first discuss the case that all variables $X,Y$ lie on lines parallel to the one $i\theta$ lies. If we are in one of the two cases of $(\textnormal{Q}3)$, the condition that the values of $\varphi$ lie on a line imply that either $\theta$ or $i\theta$ lies on the real line. For both $(\textnormal{Q}1)_{\delta=1}$ and $(\textnormal{Q}3)_{\delta=0}$, the same condition implies that the line of $X-Y$ passes through the origin. For $(\textnormal{Q}2)$ and $(\textnormal{Q}3)_{\delta=1}$, the same is true for the line of $X+Y$. Thus, we can assume that all $X,Y,i\theta$ lie on the same line. 

In light of the fact that Equation~(\ref{eq:Q}) and generic three field variables on the vertices of a quadrilateral of $\cG$ determine the fourth field variable uniquely, it is natural to demand that for a generic integrable labeling of $E(\cG)$, generic values of $X$, and generic values of all but one $Y_k$, $e_k=(x,y_k)$ an edge of $\cG_W$, the discrete Laplace Equation~(\ref{eq:Laplace}) determines the remaining $Y_s$ uniquely. This is not fulfilled in the cases of $(\textnormal{Q}2)$ and $(\textnormal{Q}3)_{\delta=1}$ if $X,Y,i\theta$ lie on the same line, using that in this setting, $\Phi(X,\cdot;\theta)$ generally cannot attain a certain range of positive values. The same is true for $(\textnormal{Q}3)_{\delta=0}$ if $X,Y,i\theta \in \mR$.

For the remaining two cases of $(\textnormal{Q}1)$ and $(\textnormal{Q}3)_{\delta=0}$ if $X,Y,i\theta \in \mR i$, we observe that $\varphi$ has a singularity if $X-Y=i\theta$. But to apply methods of real variational calculus, we need to restrict the variables to a domain where $\varphi$ can be defined smoothly. Depending on the labels $\theta$, $X-Y$ might take values only on a rather small interval, eventually contradicting the above condition that the discrete Laplace Equation~(\ref{eq:Laplace}) determines the remaining value $Y_s$ uniquely in the generic case.

Therefore, $X,Y,i\theta$ do not all lie on parallel lines. Now, if both $X-Y$ and $X+Y$ avoid $\varepsilon$-balls of $i\theta+2k\pi i$, $k\in \mZ$, $\Phi$ is not only smooth for all $X,Y$, but also bounded in the cases of $(\textnormal{Q}1)_{\delta=1}$, $(\textnormal{Q}2)$, and $(\textnormal{Q}3)$. Thus, the values of $\Phi$ neither form a spiral, nor the real line in this case.

Let us fix $X$. If the values of $\Phi$ lie on a spiral or the real line, $X-Y$ or $X+Y$ approaches the set of $i\theta+2k\pi i$, $k \in \mZ$, for any $\theta$. As above, we want to avoid singularities of $\varphi$, i.e., that $X-Y$ or $X+Y$ attains the set of $i\theta+2k\pi i$, $k\in\mZ$, for too many $\theta$s. In particular, there are $\theta$s such that no $i\theta+2k\pi i$, $k\in \mZ$, is attained. Noting that $\varphi$ is locally invertible, $X-Y$ or $X+Y$ actually approaches some fixed $i\theta+2k\pi i$ for $Y \to \infty$ or $Y \to -\infty$. Changing $\theta$ a tiny bit then gives a contradiction. As a consequence, the values of $\Phi$ indeed lie on a circle centered at the origin.

An explicit calculation for $(\textnormal{Q}1)_{\delta=1}$ shows that this happens only if $X-Y$ lies on the same line as $\theta$. Hence, all reality conditions for $(\textnormal{Q}1)_{\delta=1}$ are given by the same as for $(\textnormal{Q}1)_{\delta=0}$. By appropriate scaling and shifting, we can again restrict to $X,Y,\theta \in \mR$.

For $(\textnormal{Q}3)_{\delta=0}$, we get that $\theta$ is always real or purely imaginary, and $X \in \mR + k\pi i$, $Y \in \mR + k'\pi i$ in the first and $X,Y \in \mR$ in the second case, for some integers $k,k'$. Again, there are just three essentially different reality conditions, namely $X,\theta \in \mR$ and either $Y$ or $Y+\pi i$ is real, or $X,Y,\theta \in \mR i$.

Using that the multiplicative three-leg function of $(\textnormal{Q}2)$ or $(\textnormal{Q}3)_{\delta=1}$ is (up to a factor of $\exp(i\theta)$ in the second case) just the one of $(\textnormal{Q}1)_{\delta=1}$ or $(\textnormal{Q}3)_{\delta=0}$, respectively, multiplied with the corresponding function with $X-Y$ replaced by $X+Y$, the same reality conditions that we obtained for $(\textnormal{Q}1)_{\delta=1}$ are applicable to $(\textnormal{Q}2)$; and for $(\textnormal{Q}3)_{\delta=1}$, we can take the same conditions as for $(\textnormal{Q}3)_{\delta=0}$.

Let us now exclude further reality conditions for $(\textnormal{Q}2)$ and $(\textnormal{Q}3)_{\delta=1}$. For this, we fix some $\theta \neq 0$ such that both $X-Y$ and $X+Y$ avoid $i\theta$. The existence of such a $\theta$ follows from our consideration above. Again by an explicit calculation using that the image of $\varphi$ lies on a line parallel to the imaginary axis, $\theta$ is either real or purely imaginary in the case of $(\textnormal{Q}3)_{\delta=1}$. 

Now, it is not hard to see that the radius of the circle on which the values of $\Phi$ lie has to be one. Then, we can invert $\varphi(\cdot,\cdot;\theta)$ locally and get a two-dimensional submanifold of $\mC^2$ on which $(X,Y)$ lies. But we already know that the image of $\varphi$ of the planes $(X,Y)$ coming from the reality conditions of $(\textnormal{Q}1)_{\delta=1}$ and $(\textnormal{Q}3)_{\delta=0}$ lie on the imaginary line. In particular, these planes are the corresponding preimages and there are no further reality conditions.

In the case of $(\textnormal{Q}4)$, the situation is more complicated due to the theta functions appearing in the three-leg function. But we believe that for rectangular lattices with a half-period ratio of $-t/(2\pi i)$, $t\geq 0$, a similar reasoning as above can show that in addition to the reality conditions for $(\textnormal{Q}3)$, the only remaining reality condition is $X,\theta \in \mR i$ and $Y+t$ is purely imaginary.


\section{Generalized discrete action functionals}\label{sec:action}

This section is devoted to the study of the generalized discrete action functionals defined in (\ref{eq:action}) of discrete Laplace-type equations as in (\ref{eq:Laplace}). We stress again that in general, Equation~(\ref{eq:Laplace}) does not have to be induced by Equations~(\ref{eq:Q}) on the faces of $\cG$, although the three-leg functions come from integrable quad-equations. In particular, the labeling of $E(\cG_W)$ does not need to be integrable and the extra terms $\Theta \in \mR$ might be arbitrary. But if we consider a discrete Laplace-type integrable equation, the labeling of $E(\cG_W)$ is integrable and $\Theta_x \in 2\pi\mZ$ for inner white vertices $x$.

Having a quick look at the three-leg functions given in the appendix, we see that $\theta=0$ implies $\varphi \equiv 0$ in the case of $(\textnormal{Q}1)_{\delta=0}$ and $\Phi \equiv 1$ for the other quad-equations; and $\theta\equiv 0 \mod \pi$ implies $\Phi \equiv 1$ for $(\textnormal{Q}3)$ and $(\textnormal{Q}4)$. In regard to the discrete Laplace-type equation, we may just delete edges $e$ with $\theta_e=0$ or $\theta_e\equiv 0 \mod \pi$, respectively. Then, $\cG_W$ might split into several components, and each component can be treated separately. For simplicity, let us assume that already for all edges $e\in E(\cG_W)$, $\theta_e\neq 0$ or $\theta_e\not\equiv 0 \mod \pi$, depending on the discrete integrable quad-equation. By the same reason, we additionally exclude $\theta \in 2\pi\tau\mZ$ in the case of $(\textnormal{Q}4)$, where $\tau$ is the half-period ratio of the Jacobi theta functions.

To obtain a real functional $S$, we restrict the variables and labels to the reality conditions discussed in Section~\ref{sec:reality}. Now, we can ask for convexity or concavity of $S$, factoring out the invariance under the global variable transformation $X\mapsto X+h$ in the cases of $(\textnormal{Q}1)$ and $(\textnormal{Q}3)_{\delta=0}$. In Section~\ref{sec:main}, we summarize the main theorem of our subsequent consideration. We are able to give necessary and sufficient conditions on labels $\theta$ for strict convexity or concavity if the three-leg functions correspond to one of the integrable quad-equations $(\textnormal{Q}1)$, $(\textnormal{Q}2)$, and $(\textnormal{Q}3)$. These conditions will demand that $\theta$ or $i\theta$ is either always positive or always negative or that $\exp(i\theta)$ always lies in one of the semicircles $S^1_\pm :=\left\{z\in \mC | \left|z\right|=1, \pm \im(z) > 0 \right\}$. In the case of a discrete Laplace-type equation corresponding to $(\textnormal{Q}4)$, we find at least sufficient conditions for strict convexity or concavity.

The actual proof of Theorem~\ref{th:main} is given in Sections~\ref{sec:Q10} to~\ref{sec:Q4}. In Sections~\ref{sec:Q30} and~\ref{sec:Q31}, we will see a remarkable analogy between the generalized discrete action functionals of $(\textnormal{Q}3)$ and the Euclidean and hyperbolic circle pattern functionals of the first author and Springborn in \cite{BSp04}. Finally, we show in Section~\ref{sec:limit} how the functionals of $(\textnormal{Q}3)_{\delta=1}$ can be obtained as certain limits of the ones of $(\textnormal{Q}4)$.


\subsection{Main theorem} \label{sec:main}

We summarize the main result that will be proven in the following paragraphs. In the case of discrete Laplace-type equations corresponding to $(\textnormal{Q}4)$, we restrict to rectangular lattices, corresponding to a purely imaginary half-period ratio $\tau=-t/(2\pi i)$; the case of rhombic lattices ($|\tau|=1$) will be discussed separately in Section~\ref{sec:Q4}.

\begin{theorem}\label{th:main}
Let $\cG$ be a finite connected strongly regular bipartite quad-graph embedded in an oriented surface, and let $\cG_W$ denote the white graph corresponding to a $2$-coloring of $V(\cG)$. Moreover, let $\theta_e$ be a labeling of the edges $e \in E(\cG_W)$.

For any discrete Laplace-type equation corresponding to a discrete integrable quad-equation of the list Q, we consider reality conditions under which the corresponding generalized discrete action functional is real. To be more precise, we restrict the variables $X$ and labels $\theta$ to either the real or the imaginary line, assuming that $\theta$ is never zero (and $\theta \not\equiv 0 \mod \pi$ in the cases of $(\textnormal{Q}3)$ and $(\textnormal{Q}4)$). Then, $S$ is strictly concave on $\mR^{\left|V(\cG_W)\right|}$ or $i\mR^{\left|V(\cG_W)\right|}$, or on the subspaces $U$ or $iU$ defined by
\begin{equation}\label{eq:U}
 U=\left\{\left\{X\right\}_{x \in V(\cG_W)} \subset \mR^{\left|V(\cG_W)\right|} | \sum\limits_{x \in V(\cG_W)} X = 0\right\},
\end{equation} 
if the following restrictions of the labels $\theta_e$, $e \in E(\cG_W)$ are satisfied:
\begin{center}
\begin{tabular}{|c|c|c|c|} \hline
\textbf{quad-equation} & \textbf{functional} & \textbf{space} & \textbf{concavity condition}\\
\hline
$(\textnormal{Q}1)_{\delta=0}$ & (\ref{eq:action_Q10}) & $U$ & $\theta_e \in \mR^+$ \\
\hline
$(\textnormal{Q}1)_{\delta=1}$ & (\ref{eq:L_Q11}) & $U$ & $\theta_e \in \mR^+$ \\
\hline
$(\textnormal{Q}2)$ & (\ref{eq:L_Q2}) & $\mR^{\left|V(\cG_W)\right|}$ & $\theta_e \in \mR^+$ \\
\hline
$(\textnormal{Q}3)_{\delta=0}$ & (\ref{eq:action_Q30}) & $U$ & $\exp(i\theta_e) \in S^1_+$ \\
\hline
$(\textnormal{Q}3)_{\delta=0}$ & (\ref{eq:action_Q30E}) & $iU$ & $\theta_e \in i\mR^+$ \\
\hline
$(\textnormal{Q}3)_{\delta=1}$ & (\ref{eq:action_Q31}) & $\mR^{\left|V(\cG_W)\right|}$ & $\exp(i\theta_e) \in S^1_+$ \\
\hline
$(\textnormal{Q}3)_{\delta=1}$ & (\ref{eq:action_Q31E}) & $i\mR^{\left|V(\cG_W)\right|}$ & $\theta_e \in i\mR^+$ \\
\hline
$(\textnormal{Q}4)$ & (\ref{eq:L_Q4}) & $\mR^{\left|V(\cG_W)\right|}$ & $\exp(i\theta_e) \in S^1_+$ \\
\hline
$(\textnormal{Q}4)$ & (\ref{eq:L_Q4E}) & $i\mR^{\left|V(\cG_W)\right|}$ & $\theta_e \in i(0,r(t))$. \\
\hline
\end{tabular}
\end{center}
Here, $r(t)$ is a positive number depending on $t$. In the cases of $(\textnormal{Q}1)$, $(\textnormal{Q}2)$, and $(\textnormal{Q}3)$, the conditions on $\theta$ for concavity are even necessary conditions.

Replacing $\theta_e$ by $-\theta_e$ yields the same result with concavity replaced by convexity.
\end{theorem}

Note that the restriction to the subspace $U$ is possible if and only if \[\sum\limits_{x\in V(\cG_W)} \Theta_x=0\] in the case of $(\textnormal{Q}1)_{\delta=0}$ and in the case of the second reality condition for $(\textnormal{Q}3)_{\delta=0}$; \[\sum\limits_{x\in V(\cG_W)} \Theta_x=2\pi \sum\limits_{e\in E(\cG_W)}  \text{sgn}(\theta_e)\] in the case of $(\textnormal{Q}1)_{\delta=1}$; and \[\sum\limits_{x\in V(\cG_W)} \Theta_x=\sum\limits_{e\in E(\cG_W)} 2 \theta_e^*\] with $\theta_e^*:=\pi-\theta_e$ in the case of the first reality condition for $(\textnormal{Q}3)_{\delta=0}$.

\begin{remark}
For discrete Laplace-type equations corresponding to $(\textnormal{Q}3)$ and $(\textnormal{Q}4)$, there is another a priori different reality condition that turns out to give almost the same generalized discrete action functionals. We partition the vertex set of $V(\cG_W)$ into any two subsets $V_1$ and $V_2$. Then, $X\in \mR$ if $x \in V_1$, $X \in \mR - \pi i$ if $x \in V_2$, and $\theta \in \mR$ is also a reality condition under which the generalized discrete action functional is real. But the functional in this case is essentially the same as the corresponding Functional~(\ref{eq:action_Q30}), (\ref{eq:action_Q31}), or (\ref{eq:L_Q4}) given in Theorem~\ref{th:main}: One just replaces $L(X,Y;\theta)$ by $L(X',Y';\theta)$ if both $x$ and $y$ belong to the same vertex set and $L(X,Y;\theta)$ by $L(X',Y';\theta')$ if $x$ and $y$ belong to different vertex sets. Here, $\theta':=\theta+\pi$, $X':=X$ if $x \in V_1$, and $X':=X+\pi i$ if $x \in V_2$. 

Then, the corresponding concavity conditions are given by $\exp(i\theta_e) \in S^1_+$ if the vertices $x,y$ of $e$ belong to the same vertex set and $\exp(i\theta_e) \in S^1_-$ otherwise.

In the case of $(\textnormal{Q}4)$, a similar reasoning applies for the reality condition $X\in \mR i$ if $x \in V_1$, $X \in \mR i -t$ if $x \in V_2$, and $\theta \in \mR i$. Then, $\theta':=\theta+it$, and the concavity condition transforms to $\theta_e \in (-t,-t+r(t))i$ for those edges $e$ whose vertices lie in different vertex sets.
\end{remark}


\subsection{\texorpdfstring{$(\textnormal{Q}1)_{\delta=0}$}{(Q1) -- delta=0}} \label{sec:Q10}

We have no change of variables, i.e., $X=x$, $Y=y$, etc. We will investigate the reality condition $X,Y,\theta_e \in \mR$ for all edges $e=(x,y)\in E(\cG_W)$, assuming that $\theta_e$ is never zero. The three-leg function $\varphi$ is given by
\begin{equation*}
 \varphi(X,Y;\theta)=\frac{\theta}{X-Y},
 \end{equation*}
which is smooth as long $X \neq Y$. The generalized discrete action functional is given by
\begin{equation}\label{eq:action_Q10}
 S=\sum\limits_{e=(x,y)\in E(\cG_W)} \theta_e \log \left|X-Y\right| - \sum\limits_{x \in V(\cG_W)} \Theta_x X.
 \end{equation}
Here, $\log$ is any fixed branch of the logarithm. $S$ is smooth as long $X \neq Y$ for all edges $e=(x,y)$; otherwise, $S$ is not defined. The Hessian of $S$ can be easily calculated:
\begin{equation*}
 D^2S=\sum\limits_{e=(x,y)\in E(\cG_W)} \frac{-\theta_e}{(X-Y)^2} (dX-dY)^2.
 \end{equation*}

Now, we can expect strict convexity or concavity only on the subspace $U$ as defined in (\ref{eq:U}). But to be able to restrict to $U$, $S$ has to be invariant under the global variable transformation $X \mapsto X+h$, $h\in\mR$. This is the case if and only if
\begin{equation*}
 \sum\limits_{x\in V(\cG_W)} \Theta_x=0
\end{equation*}
is satisfied. Clearly, strict concavity on $U$ is obtained if $\theta_e \in \mR^+$ for all $e\in E(\cG_W)$ and strict convexity if $\theta_e \in \mR^-$ for all $e\in E(\cG_W)$.

Conversely, let us assume that $S$ is concave (or convex). For an edge $e$ with vertices $x$ and $y$, let us consider $S$ as a function of $X$ and $Y$. The other $Z$, $z \in V(\cG_W) \backslash \left\{x,y\right\}$, are just parameters. If $Z \to \pm \infty$ for all other vertices $z$, \begin{equation*} D^2S(X,Y) \to \frac{-\theta_e}{(X-Y)^2} (dX-dY)^2. \end{equation*} In particular, $S$ is concave (or convex) only if $\theta_e > 0$ (or $\theta_e < 0$).


\subsection{\texorpdfstring{$(\textnormal{Q}1)_{\delta=1}$}{(Q1) -- delta=1}} \label{sec:Q11}

As before, we have no change of variables and we investigate the reality condition $X,Y,\theta_e \in \mR$ for all edges $e=(x,y)\in E(\cG_W)$, $\theta_e \neq 0$. The three-leg function $\Phi$ is given by
\begin{equation*}
 \Phi(X,Y;\theta)=\frac{X-Y+i\theta}{X-Y-i\theta}.
 \end{equation*}
Let us choose
\begin{equation*}
 \varphi(X,Y;\theta)=-i\log(X-Y+i\theta)+i\log(X-Y-i\theta),
 \end{equation*}
where we consider $\log$ to be the principle branch of the logarithm. One can easily check that
\begin{align}
L(X,Y;\theta)&=(iX-iY+\theta)\log(X-Y-i\theta) \notag\\
&-(iX-iY-\theta)\log(X-Y+i\theta) + 2\pi Y \text{sgn}(\theta) \label{eq:L_Q11}
\end{align}
satisfies the Conditions~(\ref{eq:symmetry})~and~(\ref{eq:derivative}) and $L(X,Y;\theta)=-L(X,Y;-\theta)$. Then, the generalized discrete action functional $S$ is just given by Formula~(\ref{eq:action}) and $S$ is smooth everywhere. An elementary calculation shows that its Hessian is
\begin{equation*}
 D^2S= \sum\limits_{e=(x,y)\in E(\cG_W)} \frac{-2\theta_e}{(X-Y)^2+\theta_e^2} (dX-dY)^2.
 \end{equation*}
 
Again, we can only expect strict concavity or convexity if we restrict to the subspace $U$ as defined in (\ref{eq:U}). With the same arguments as in Section~\ref{sec:Q10}, $S$ is strictly concave if and only if $\theta_e \in \mR^+$ for all $e\in E(\cG_W)$ and strictly convex if and only if $\theta_e \in \mR^-$ for all $e\in E(\cG_W)$.

To be able to restrict to the subspace $U$, the functional $S$ has to be invariant under the global variable transformation $X \mapsto X+h$, $h\in\mR$, such that
\begin{equation}\label{eq:scaling_Q11}
\sum\limits_{x\in V(\cG_W)} \Theta_x=2\pi \sum\limits_{e\in E(\cG_W)}  \text{sgn}(\theta_e).
\end{equation}
 
Note that
\begin{align*}
0&< \varphi(X,Y;\theta)< \ \ \, 2 \pi \text{ if } \theta>0,\\
0&> \varphi(X,Y;\theta)> -2 \pi \text{ if } \theta<0.
\end{align*}
So according to the discrete Laplace-type Equation~(\ref{eq:Laplace}),
\begin{align*}
\forall x \in V(\cG_W): 0&<\Theta_x< \ \ \, 2 \pi \deg_{\cG_W}(x) \text{ if } \theta_e>0 \; \forall e \in E(\cG_W),\\
\forall x \in V(\cG_W): 0&>\Theta_x> -2 \pi \deg_{\cG_W}(x) \text{ if } \theta_e<0 \; \forall e \in E(\cG_W)
\end{align*}
is necessary to obtain solutions. Here, $\deg_{\cG_W}(x)$ denotes the \textit{degree} of $x$ in $\cG_W$, i.e., the number of adjacent white vertices.


\subsection{\texorpdfstring{$(\textnormal{Q}2)$}{(Q2)}} \label{sec:Q2}

Now, we change the variables by $x=X^2$, $y=Y^2$, etc. The reality condition is given by $X,Y, \theta_e \in \mR$ for all edges $e=(x,y)\in E(\cG_W)$, $\theta_e \neq 0$. In the case of $(\textnormal{Q}2)$, the three-leg function $\Phi$ is given by
\begin{equation*}
 \Phi(X,Y;\theta)=\frac{(X-Y+i\theta) (X+Y+i\theta)}{(X-Y-i\theta) (X+Y-i\theta)}.
 \end{equation*}
We choose
\begin{equation*}
 i\varphi(X,Y;\theta)=\log(X-Y+i\theta)-\log(X-Y-i\theta)+\log(X+Y+i\theta)-\log(X+Y-i\theta),
 \end{equation*}
where $\log$ is the principle branch of the logarithm. Then, a function $L$ satisfying Conditions~(\ref{eq:symmetry})~and~(\ref{eq:derivative}) and $L(X,Y;\theta)=-L(X,Y;-\theta)$ is given by
\begin{align}
L(X,Y;\theta)&=(iX-iY+\theta)\log(X-Y-i\theta)-(iX-iY-\theta)\log(X-Y+i\theta) \notag\\
&+(iX+iY+\theta)\log(X+Y-i\theta)-(iX+iY-\theta)\log(X+Y+i\theta) \notag\\
&+2\pi Y \text{sgn}(\theta). \label{eq:L_Q2}
\end{align}
Now, the generalized discrete action functional is given by Formula~(\ref{eq:action}). Using the result of Section~\ref{sec:Q11},
\begin{align*}
 D^2S&=\sum\limits_{e=(x,y)\in E(\cG_W)}\left( \frac{-2\theta_e}{(X-Y)^2+\theta_e^2} (dX-dY)^2 +\frac{-2\theta_e}{(X+Y)^2+\theta_e^2} (dX+dY)^2\right).
\end{align*}

Clearly, $S$ is strictly concave (or strictly convex) on $\mR^{\left| V(\cG_W)\right|}$ if $\theta_e \in \mR^+$ (or $\theta_e \in \mR^- $) for all $e\in E(\cG_W)$. With a similar reasoning as in Section~\ref{sec:Q11}, these conditions are even necessary for concavity (or convexity) and
\begin{align*}
\forall x \in V(\cG_W): 0&<\Theta_x< \ \ \, 4 \pi \deg_{\cG_W}(x) \text{ if } \theta_e>0 \; \forall e \in E(\cG_W),\\
\forall x \in V(\cG_W): 0&>\Theta_x> -4 \pi \deg_{\cG_W}(x) \text{ if } \theta_e<0 \; \forall e \in E(\cG_W)
\end{align*}
is necessary to obtain solutions of the discrete Laplace-type equation.


\subsection{\texorpdfstring{$(\textnormal{Q}3)_{\delta=0}$}{(Q3) -- delta=0}} \label{sec:Q30}

We consider the following change of variables: $x=\exp(X)$, $y=\exp(Y)$, etc. Let us start with the reality condition $X,Y,\theta_e \in \mR$ for all $e=(x,y)\in E(\cG_W)$. In addition, we assume that $\theta_e \not\equiv 0 \mod \pi$. Since the quad-equation does not change if we add multiples of $2\pi$ to $\alpha$ or $\beta$, $\theta$ is defined only up to a multiple of $2 \pi$. So let us assume that all labels $\theta \in (0,2 \pi)$. A change of the interval might only influence the terms $\Theta$ in the discrete Laplace-type Equation~\ref{eq:Laplace}, but not the following arguments.

The three-leg function $\Phi$ is given by
\begin{equation}\label{eq:Phi_Q30}
 \Phi(X,Y;\theta)=\exp(-i \theta)\frac{\sinh(\frac{X-Y+i\theta}{2})}{\sinh(\frac{X-Y-i\theta}{2})},
 \end{equation}
so let us choose
\begin{equation}\label{eq:phi_Q30} \varphi(X,Y;\theta)=-i\log\left(-\frac{\sinh(\frac{X-Y+i\theta}{2})}{\sinh(\frac{X-Y-i\theta}{2})}\right)+(\pi-\theta),
 \end{equation}
where $\log$ is the principle branch of the logarithm. $\varphi$ is smooth everywhere.

To find a symmetric primitive of $\varphi$, let us introduce the dilogarithm \[\Li(z)= -\int\limits_0^z \frac{\log(1-u)}{u}du.\] The dilogarithm is an analytic function on the complex plane cut along $[1,\infty)$. Following the presentation of the first author and Springborn in \cite{BSp04}, we have \begin{align*}\im \Li(\exp(X-Y+i\theta))&=\frac{1}{2i}\Li(\exp(X-Y+i\theta))-\frac{1}{2i}\Li(\exp(X-Y-i\theta))\\&=\frac{1}{2i}\int\limits_{-\infty}^{X-Y}\log\left(\frac{1-\exp(v-i\theta)}{1-\exp(v+i\theta)}\right) dv,\end{align*} substituting $u=\exp(v+i\theta)$ into the integral representation of the dilogarithm.

Observing that \[\frac{\sinh(\frac{X-Y+i\theta}{2})}{\sinh(\frac{X-Y-i\theta}{2})}=\exp(-i\theta)\frac{1-\exp(X-Y+i\theta)}{1-\exp(X-Y-i\theta)}=\exp(i\theta)\frac{1-\exp(Y-X-i\theta)}{1-\exp(Y-X+i\theta)},\] we derive the generalized discrete action functional
\begin{align}
 S&=\sum\limits_{e=(x,y)\in E(\cG_W)} -\left(\im \Li(\exp(X-Y+i \theta_e))+\im \Li(\exp(Y-X+i \theta_e))\right)\notag\\
&+\sum\limits_{e=(x,y)\in E(\cG_W)}\theta_e^* (X+Y)-\sum\limits_{x\in V(\cG_W)}\Theta_x X, \label{eq:action_Q30}
\end{align}
where $\theta^*:= \pi - \theta$. Note that $\theta_e^*$ can be replaced by $-\theta_e$ if $\Theta$s are changed appropriately, leading to a primitive $L$ that satisfies $L(X,Y;\theta)=-L(X,Y;-\theta)$.

$-S$ coincides with the Euclidean circle pattern functional of the first author and Springborn in \cite{BSp04}, describing a Euclidean circle pattern combinatorially equivalent to $\cG_W$ with logarithmic radii $X$, interior intersection angles $\theta_e^* \in (0,\pi)$, and cone (or boundary) angles $\Theta_x>0$. We will investigate this interpretation again in Section~\ref{sec:circle_patterns}.

An elementary calculation shows the same result for the Hessian of $S$ as in \cite{BSp04},
\begin{equation}
 D^2S=\sum\limits_{e=(x,y)\in E(\cG_W)} \frac{-\sin(\theta_e)}{\cosh(X-Y)-\cos(\theta_e)} (dX-dY)^2.\label{eq:D2_Q30}
\end{equation}

Clearly, we can only expect strict concavity or convexity if we restrict the variables $X$ to the subspace $U$ as defined in (\ref{eq:U}). Then, $S$ is strictly concave if and only if $\theta_e \in (0,\pi)$ for all $e\in E(\cG_W)$ and strictly convex if and only if $\theta_e \in (\pi,2\pi)$ for all $e\in E(\cG_W)$. Since $\theta$ was a priori only defined up to a multiple of $2\pi$, we can state more generally that $S$ is strictly concave if and only if $\exp(i\theta_e) \in S^1_+$ for all $e\in E(\cG_W)$ and strictly convex if and only if $\exp(i\theta_e) \in S^1_-$ for all $e\in E(\cG_W)$.

To be able to restrict to the subspace $U$, the functional $S$ has to be invariant under the global variable transformation $X \mapsto X+h$, requiring
\begin{equation}\label{eq:scaling_Q30}
 \sum\limits_{x\in V(\cG_W)} \Theta_x=\sum\limits_{e\in E(\cG_W)} 2 \theta_e^*.
\end{equation}

Let $s\in (0,\pi)$. Then,
\begin{equation}\label{eq:remark_Q30}
   \varphi(X,Y; \pi + s)=-\varphi(X,Y; \pi - s).
\end{equation}
In particular, the discrete Laplace-type Equations~(\ref{eq:Laplace}) with $\theta_e=\pi -s_e$ and $\tilde{\theta}_e=\pi +s_e$ are equivalent if one chooses $\tilde{\Theta}_x=-\Theta_x$. Also, $L(X,Y;\pi -s)=-L(X,Y;\pi+s)$ as it appears in Functional~(\ref{eq:action_Q30}). Thus, we may restrict our attention to the concavity condition $\exp(i\theta_e) \in S^1_+$ for all $e\in E(\cG_W)$. 

In this case, $0<\varphi(X,Y;\theta)<2\pi$, so $0<\Theta_x<2 \pi \deg_{\cG_W} (x)$ for all vertices $x$ is necessary to obtain solutions. Therefore, the terms $\Theta_x$ can be indeed seen as cone (or boundary) angles of a Euclidean circle pattern and Functional~(\ref{eq:action_Q30}) can be identified with the Euclidean circle pattern functional. Also, Condition~\ref{eq:scaling_Q30} appeared in the first author's and Springborn's paper \cite{BSp04} and it will come again in Section~\ref{sec:Q_Dirichlet}. There, we will discuss when maxima of $S$ exist and adapt the results obtained by the first author and Springborn to our setting.

\begin{remark}
Let us shortly discuss the reality condition described by $X\in \mR$ if $x \in V_1$, $X \in \mR - \pi i$ if $x \in V_2$, and $\theta \in \mR \backslash \mZ\pi$, where $V_1 \dot{\cup} V_2$ is some partition of $V(\cG_W)$ into two vertex sets. Now, define $X':=X$ if $x \in V_1$, $X':=X+\pi i$ if $x \in V_2$, $\theta'_e:=\theta_e$ if the vertices of the edge $e$ lie in the same vertex set, and $\theta'_e:=\theta_e +\pi$ if the vertices of $e$ are in different vertex sets. Then,
\begin{equation*}
 \Phi(X,Y;\theta)=\exp(-i \theta^\prime)\frac{\sinh(\frac{X^\prime-Y^\prime+i\theta^\prime}{2})}{\sinh(\frac{X^\prime-Y^\prime-i\theta^\prime}{2})}.
\end{equation*}
Therefore, we can argue exactly as above with variables $X'$ instead of $X$ and parameters $\theta^\prime$ instead of $\theta$. In particular, the generalized discrete action functional is strictly concave if and only if $\exp(i\theta_e) \in S^1_+$ if the vertices $x,y$ of $e$ belong to the same vertex set and $\exp(i\theta_e) \in S^1_-$ otherwise.

Note that due to Equation~(\ref{eq:remark_Q30}), another way to represent the generalized discrete action functional under this modified reality condition is to replace $L(X,Y;\theta_e)$ by $L(X',Y';\theta_e)$ if the vertices of the edge $e$ belong to the same vertex set and $L(X,Y;\theta_e)$ by $-L(X',Y';\theta_e^*)$ otherwise.
\end{remark}

Let us now come to the reality condition $X,Y,\theta \in \mR i$, $\theta \neq 0$. We introduce new variables $X'$ and labels $\theta'$ by $X=iX'$, $\theta=i\theta'$ and consider the variational formulation with respect to these new variables. Then, we get for the three-leg function 
\begin{equation*}
 \Phi'(X,Y;\theta)=\frac{\sin(\frac{X'-Y'+i\theta'}{2})}{\sin(\frac{X'-Y'-i\theta'}{2})}
 \end{equation*}
and we can choose
\begin{equation*} \varphi(X,Y;\theta)=-i\log\left(\frac{\sin(\frac{X'-Y'+i\theta'}{2})}{\sin(\frac{X'-Y'-i\theta'}{2})}\right).
 \end{equation*}
Note that we do not take the three-leg function $\Phi$, but $\Phi'$ instead. The reason for taking $\Phi$ rather than $\Phi'$ before was to identify the generalized discrete action functional with the Euclidean circle pattern functional. Taking $\Phi$ instead of $\Phi'$ under the actual reality condition would give an additional imaginary summand of $-i\theta'$ to $\varphi$ that we want to avoid. Furthermore, we cannot take one branch of the logarithm to define $\varphi$ globally. But it is easy to see that there is a smooth extension of $\varphi$ to all $X',Y' \in \mR$. We just have to fix one value, and we take $\varphi(X',X';\theta)=\pi$.

In a similar way as before, we construct a symmetric primitive of $\varphi$, but we have to take into account that there is not just one branch of the logarithm in the definition of $\varphi$. Substituting $u=\exp(w+\theta')$ into the integral representation of the dilogarithm yields
\begin{align*}
-\Li(\exp(iX'-iY'+\theta'))&=\int\limits_{-\infty}^{i(X'-Y')}\log(1-\exp(w+\theta'))dw\\&=\int\limits_{-\infty}^{0}\log(1-\exp(w+\theta'))dw\\&+i\int\limits_{0}^{X'-Y'}\log(1-\exp(iw+\theta'))dw.
\end{align*}
For $\theta'<0$, $\Li(\exp(iX'-iY'+\theta'))$ is just the ordinary dilogarithm, but for $\theta'>0$, we extend $\Li(\exp(iX'-iY'+\theta'))$ to a smooth function of $X',Y'\in \mR$ by the integral representation above, choosing appropriate branches of the logarithm during integration. We demand that for $X'=Y'$, the extended dilogarithm coincides with the ordinary. Again, \begin{align*}\frac{\sin(\frac{X'-Y'+i\theta'}{2})}{\sin(\frac{X'-Y'-i\theta'}{2})}&=\exp(\theta')\frac{1-\exp(i(X'-Y')-\theta')}{1-\exp(i(X'-Y')+\theta')}\\&=\exp(-\theta')\frac{1-\exp(i(Y'-X')+\theta')}{1-\exp(i(Y'-X')-\theta')},\end{align*} and we derive the generalized discrete action functional
\begin{align}
 S&=\frac{1}{2}\sum\limits_{e=(x,y)\in E(\cG_W)} \left( \Li(\exp(iX'-iY'- \theta'_e))-\Li(\exp(iX'-iY'+ \theta'_e))\right)\notag\\
 &+\frac{1}{2}\sum\limits_{e=(x,y)\in E(\cG_W)} \left( \Li(\exp(iY'-iX'- \theta'_e))-\Li(\exp(iY'-iX'+ \theta'_e))\right)\notag\\
 &-\sum\limits_{x\in V(\cG_W)}\Theta_x X'\notag\\
&=\frac{1}{2}\sum\limits_{e=(x,y)\in E(\cG_W)} \left( \Li(\exp(X-Y+i\theta_e))-\Li(\exp(X-Y-i\theta_e))\right)\notag\\
 &+\frac{1}{2}\sum\limits_{e=(x,y)\in E(\cG_W)} \left( \Li(\exp(Y-X+i\theta_e))-\Li(\exp(Y-X-i\theta_e))\right)\notag\\
&+\sum\limits_{x\in V(\cG_W)}\Theta_x iX. \label{eq:action_Q30E}
\end{align}

Using the analogy of this functional with (\ref{eq:action_Q30E}) before, we obtain
\begin{equation*}
 D^2S=\sum\limits_{e=(x,y)\in E(\cG_W)} \frac{\sinh(\theta'_e)}{\cos(X'-Y')-\cosh(\theta'_e)} (dX'-dY')^2.
\end{equation*}

Again, we can only expect strict concavity or convexity if we restrict the variables $X'$ to the subspace $U$ as defined in (\ref{eq:U}). But now, $S$ is strictly concave if and only if $\theta'_e>0$ for all $e\in E(\cG_W)$ and $S$ is strictly convex if and only if $\theta'_e<0$ for all $e\in E(\cG_W)$. Thus, $X\in iU$ and the conditions for concavity and convexity correspond to $\theta_e \in \mR^+i$ and $\theta_e\in \mR^-i$, respectively. In contrast to the reality conditions considered above, the restriction to the subspace $iU$ is now possible if and only if
\begin{equation*}
 \sum\limits_{x\in V(\cG_W)} \Theta_x=0,
\end{equation*}
which corresponds to the same scaling condition as for $(\textnormal{Q}1)_{\delta=0}$.


\subsection{\texorpdfstring{$(\textnormal{Q}3)_{\delta=1}$}{(Q3) -- delta=1}} \label{sec:Q31}

This time, we choose $x=\cosh(X)$, $y=\cosh(Y)$, etc. as change of variables. Again, we start with the reality condition $X,Y,\theta_e \in \mR$ and assume that $\theta_e \notin \pi\mZ$ for all edges $e=(x,y)\in E(\cG_W)$. Since the formulation of the quad-equation $(\textnormal{Q}3)_{\delta=1}$ does not change if a multiple of $2\pi$ is added to a label, we may choose $\theta_e \in (0,2\pi)$ for all edges $e$.

The three-leg function is given by
\begin{equation}\label{eq:Phi_Q31}
 \Phi(X,Y;\theta)=\frac{\sinh(\frac{X-Y+i\theta}{2}) \sinh(\frac{X+Y+i\theta}{2})}{\sinh(\frac{X-Y-i\theta}{2})\sinh(\frac{X+Y-i\theta}{2})}
 \end{equation}
and we choose
\begin{equation}\label{eq:phi_Q31} i\varphi(X,Y;\theta)=\log\left(-\frac{\sinh(\frac{X-Y+i\theta}{2})}{\sinh(\frac{X-Y-i\theta}{2})}\right)+\log\left(-\frac{\sinh(\frac{X+Y+i\theta}{2})}{\sinh(\frac{X+Y-i\theta}{2})}\right),
 \end{equation}
where $\log$ is the principle branch of the logarithm. Essentially the same calculations as in Section~\ref{sec:Q30} yield
\begin{align}
S=&-\sum\limits_{e=(x,y)\in E(\cG_W)}\left(\im \Li(\exp(X-Y+i \theta_e))+\im \Li(\exp(Y-X+i \theta_e))\right)\notag\\
&-\sum\limits_{e=(x,y)\in E(\cG_W)} \left(\im \Li(\exp(X+Y+i \theta_e))+\im \Li(\exp(-X-Y+i \theta_e))\right)\notag\\
&-\sum\limits_{x\in V(\cG_W)}\Theta_x X \label{eq:action_Q31}
\end{align}
as the generalized discrete action functional with Hessian
\begin{align}
 D^2S&=\sum\limits_{e=(x,y)\in E(\cG_W)}\frac{-\sin(\theta_e)}{\cosh(X-Y)-\cos(\theta_e)}(dX-dY)^2\notag\\
&+\sum\limits_{e=(x,y)\in E(\cG_W)}\frac{-\sin(\theta_e)}{\cosh(X+Y)-\cos(\theta_e)} (dX+dY)^2.\label{eq:D2_Q31}
\end{align}
It follows that $S$ is strictly concave if and only if $\exp(i\theta_e) \in S^1_+$ for all edges $e$ and strictly convex if and only if $\exp(i\theta_e) \in S^1_-$ for all $e\in E(\cG_W)$. Since $\left|\varphi(X,Y;\theta)\right|<2\pi$ if $\theta \not\equiv 0 \mod 2\pi$, $\left| \Theta_x\right|<2\pi \deg_{\cG_W}(x)$ is necessary to obtain solutions of the discrete Laplace-type Equation~(\ref{eq:Laplace}).

Now, $-S$ coincides with the hyperbolic circle pattern functional of the first author and Springborn in \cite{BSp04}, describing a hyperbolic circle pattern combinatorially equivalent to the white graph $\cG_W$ with radii $2\artanh(\exp(X))$, interior intersection angles $\theta_e^* \in (0,\pi)$, and cone (or boundary) angles $\Theta_x>0$. We refer to Section~\ref{sec:circle_patterns} below for some more details.

If $s\in (0,\pi)$,
\begin{equation}\label{eq:remark_Q31}
   \varphi(X,Y; \pi + s)=-\varphi(X,Y; \pi - s),
\end{equation}
and the discrete Laplace-type Equations~(\ref{eq:Laplace}) with $\theta_e=\pi -s_e$ and $\tilde{\theta}_e=\pi +s_e$ can be identified, changing $\tilde{\Theta}_x=-\Theta_x$. So besides $L(X,Y;\theta)=-L(X,Y;-\theta)$, we also have $L(X,Y;\theta)=-L(X,Y;\tilde{\theta})$. Hence, we can restrict to the concavity condition $\exp(i\theta_e) \in S^1_+$ for all $e\in E(\cG_W)$. Then, $\varphi(X,Y;\theta)>0$ if $X,Y<0$, so  $\Theta_x>0$ for all vertices $x$ is necessary to obtain solutions with $X<0$ for all $x\in V(\cG_W)$. As a consequence, the terms $\Theta_x$ can be again seen as cone (or boundary) angles of a circle pattern and Functional~(\ref{eq:action_Q30}) can be identified with the hyperbolic circle pattern functional.

\begin{remark}
For the reality condition $X\in \mR$ if $x \in V_1$, $X \in \mR - \pi i$ if $x \in V_2$, and $\theta \in \mR \backslash \mZ\pi$, where $V_1 \dot{\cup} V_2$ is some partition of $V(\cG_W)$ into two vertex sets, we can argue exactly as in Section~\ref{sec:Q30}. For this, observe that Equation~(\ref{eq:remark_Q31}) is the analog to Equation~(\ref{eq:remark_Q30}).
\end{remark}

Considering the reality condition $X,Y,\theta \in \mR i$, $\theta \neq 0$, we introduce new variables $X'$ and labels $\theta'$ by $X=iX'$, $\theta=i\theta'$ and consider the variational formulation with respect to these new variables. With the same reasoning as in the end of Section~\ref{sec:Q30}, we get the generalized discrete action functional
\begin{align}
S&=\frac{1}{2}\sum\limits_{e=(x,y)\in E(\cG_W)} \left( \Li(\exp(X-Y+i\theta_e))-\Li(\exp(X-Y-i\theta_e))\right)\notag\\
 &+\frac{1}{2}\sum\limits_{e=(x,y)\in E(\cG_W)} \left( \Li(\exp(Y-X+i\theta_e))-\Li(\exp(Y-X-i\theta_e))\right)\notag\\
 &+\frac{1}{2}\sum\limits_{e=(x,y)\in E(\cG_W)} \left( \Li(\exp(X+Y+i\theta_e))-\Li(\exp(X+Y-i\theta_e))\right)\notag\\
 &+\frac{1}{2}\sum\limits_{e=(x,y)\in E(\cG_W)} \left( \Li(\exp(-X-Y+i\theta_e))-\Li(\exp(-X-Y-i\theta_e))\right)\notag\\
&+\sum\limits_{x\in V(\cG_W)}\Theta_x iX \label{eq:action_Q31E}
\end{align}
with Hessian
\begin{align*}
 D^2S&=\sum\limits_{e=(x,y)\in E(\cG_W)}\frac{\sinh(\theta'_e)}{\cos(X'-Y')-\cosh(\theta'_e)} (dX'-dY')^2\\
&+\sum\limits_{e=(x,y)\in E(\cG_W)}\frac{\sinh(\theta'_e)}{\cos(X'+Y')-\cosh(\theta'_e)} (dX'+dY')^2.
\end{align*}

Again, $S$ is strictly concave if and only if $\theta_e\in \mR^+i$ for all edges $e$ of $\cG_W$ and strictly convex if and only if $\theta_e\in \mR^-i$ for all edges $e$.


\subsection{\texorpdfstring{$(\textnormal{Q}4)$}{(Q4)}} \label{sec:Q4}

In this section, we will use several facts of complex analysis and the theory of elliptic functions without going into much detail; most of them can be found in the book of Hurwitz \cite{H00}. But in contrast to the notation in his book, we define the Jacobi theta functions $\vartheta_1,\vartheta_2,\vartheta_3,\vartheta_4=\vartheta_0$ with half-period ratio $\tau$ by the following, where $h^k:=\exp(i\pi \tau k)$ for $k \in \mR$:
\begin{align*}
\vartheta_1(v)&=2\left(h^{\frac{1}{4}}\sin(v)-h^{\frac{9}{4}}\sin(3v)+h^{\frac{25}{4}}\sin(5v)\mp \ldots\right),\\
\vartheta_2(v)&=2\left(h^{\frac{1}{4}}\cos(v)+h^{\frac{9}{4}}\cos(3v)+h^{\frac{25}{4}}\cos(5v)+ \ldots\right),\\
\vartheta_3(v)&=1+2\left(h\cos(2v)+h^4\cos(4v)+h^9\cos(6v)+ \ldots\right),\\
\vartheta_4(v)&=1-2\left(h\cos(2v)-h^4\cos(4v)+h^9\cos(6v)\mp \ldots\right).
\end{align*}
A scaling of the argument by the factor $\pi$ relates these two standard formulations of Jacobi theta functions.

Coming back to the discrete integrable quad-equation $(\textnormal{Q}4)$, we consider the following change of variables: $x=\sn(\pi/2-iX)$, $y=\sn(\pi/2-iY)$, etc., where $\sn$ is the Jacobi elliptic function $\sn$ with modulus $\kappa=\vartheta_2^2(0)\vartheta_3^{-2}(0)$.

In the following, we mainly consider purely imaginary $\tau$, i.e., rectangular period lattices. In the end, we will comment on the situation of rhombic lattices.

First, we investigate the reality condition $X,Y,\theta_e \in \mR$, assuming that for all edges $e=(x,y)\in E(\cG_W)$, $\theta_e \not\equiv 0\mod  \pi$. The three-leg function $\Phi$ is given by
\begin{equation}
 \Phi(X,Y;\theta)
=\frac
{\vartheta_1\left(\frac{X-Y+i\theta}{2i}\right)
\vartheta_1\left(\frac{X+Y+i\theta}{2i}\right)}
{\vartheta_1\left(\frac{X-Y-i\theta}{2i}\right)
\vartheta_1\left(\frac{X+Y-i\theta}{2i}\right)}. \label{eq:Phi_Q4}
\end{equation}
Since $\Phi$ is $2\pi$-periodic in the $\theta$-variable, we may choose the labels $\theta \in (0,2\pi)$.

Now, fix for a moment some $\theta$. Let $U$ be an open strip along the real line such that for all $u \in U$, $\im (u \pm i\theta) \not\equiv 0 \mod 2 \pi$. Then, the function \[u \mapsto -\frac{\vartheta_1(\frac{u+i\theta}{2i})}{\vartheta_1(\frac{u-i\theta}{2i})}\] is holomorphic and has no zeroes in $U$. Moreover, it takes the value 1 if $u=0$. Thus, there is a holomorphic $s_1(u,\theta):U \to \mC$ such that $s_1(0;\theta)=0$ and \[\exp(s_1(u,\theta))=-\frac{\vartheta_1(\frac{u+i\theta}{2i})}{\vartheta_1(\frac{u-i\theta}{2i})}.\] We now define \[\log\left(-\frac{\vartheta_1(\frac{u+i\theta}{2i})}{\vartheta_1(\frac{u-i\theta}{2i})}\right):=s_1(u;\theta),\] which is a smooth and purely imaginary function for real $u$. Indeed, that $s_1(u;\theta)$ is purely imaginary follows from $\vartheta_1(\bar{v})=\overline{\vartheta_1(v)}$ due to the conjugate symmetry of the trigonometric functions. This allows us to choose
\begin{equation}
i\varphi(X,Y;\theta)=\log\left(-\frac{\vartheta_1\left(\frac{X-Y+i\theta}{2i}\right)}{\vartheta_1\left(\frac{X-Y-i\theta}{2i}\right)}\right)+\log\left(-\frac{\vartheta_1\left(\frac{X+Y+i\theta}{2i}\right)}{\vartheta_1\left(\frac{X+Y-i\theta}{2i}\right)}\right).\label{eq:phi_Q4}
\end{equation}

Now, we define the smooth function $I_1(\cdot;\theta):\mR \rightarrow \mR$ by
\begin{equation*}
I_1(u;\theta)=-i\int\limits_0^u \log\left(-\frac{\vartheta_1(\frac{u+i\theta}{2i})}{\vartheta_1(\frac{u-i\theta}{2i})}\right) du,
\end{equation*}
where we integrate along the real line. Clearly, $I_1(0;\theta)=0$, and $I_1$ is even since $\vartheta_1$ is an odd function. Therefore, 
\begin{equation}\label{eq:L_Q4}
L(X,Y;\theta)=I_1(X-Y;\theta)+I_1(X+Y;\theta)
\end{equation}
satisfies Conditions~(\ref{eq:symmetry})~and~(\ref{eq:derivative}) and $L(X,Y;\theta)=-L(X,Y;-\theta)$. Hence, the generalized discrete action functional is given by Formula~(\ref{eq:action}).

Let $\omega:=\pi i$ and $\omega^\prime:=\tau \omega =-t/2<0$. We introduce the meromorphic function
\begin{equation*}
\tilde{\zeta}(u):=\frac{\pi}{2\omega} \frac{\vartheta_1^\prime(v)}{\vartheta_1(v)}=\frac{1}{2i}\frac{d\log(\vartheta_1(v))}{dv}, \quad v=\frac{\pi u}{2 \omega}=\frac{u}{2i}.
\end{equation*}
As $\vartheta_1$, $\tilde{\zeta}$ is conjugate symmetric. Note that
\begin{equation*}
\tilde{\zeta}(u)=\zeta(u)-\frac{\eta u}{\omega},
\end{equation*}
where $\eta:=\zeta(\omega)$ and $\zeta$ is the Weierstra{\ss} zeta-function corresponding to the Weierstra{\ss} $\wp$-function with half-periods $\omega,\omega^\prime$. We eventually obtain
\begin{align}
D^2S=&-i\sum\limits_{e=(x,y)\in E(\cG_W)} \left(\tilde{\zeta}(X-Y+i\theta_e)-\tilde{\zeta}(X-Y-i\theta_e)\right) (dX-dY)^2 \notag\\
&-i\sum\limits_{e=(x,y)\in E(\cG_W)}\left(\tilde{\zeta}(X+Y+i\theta_e)-\tilde{\zeta}(X+Y-i\theta_e)\right) (dX+dY)^2. \label{eq:D2_Q4}
\end{align}

As above, let $h^k:=\exp(i \pi \tau k)=\exp(-tk/2)$ for any $k \in \mR$. Moreover, define $z^{\pm2}:=\exp(\pm u)$. Then,
\begin{align}
\tilde{\zeta}(u)&=\frac{\pi}{2\omega} \cot\left(\frac{\pi u}{2\omega}\right)+\frac{i\pi}{\omega} \sum\limits_{k=1}^\infty \left(\frac{h^{2k}z^{-2}}{1-h^{2k}z^{-2}}-\frac{h^{2k}z^2}{1-h^{2k}z^2}\right)\notag\\
&=-\frac{i}{2} \cot\left(-\frac{i u}{2}\right)+\sum\limits_{k=1}^\infty \frac{\exp(-tk)\exp(-u)-\exp(-tk)\exp(u)}{\exp(-2tk)-\exp(-tk)(\exp(u)+\exp(-u))+1}\notag\\
&=\frac{1}{2} \coth\left(\frac{u}{2}\right)+\sum\limits_{k=1}^\infty \frac{\sinh(u)}{\cosh(u)-\cosh(tk)}. \label{eq:zeta_sum}
\end{align}
Series~(\ref{eq:zeta_sum}) converges absolutely if the denominator never vanishes. In particular, the series converges for $u=Z\pm i\theta$ if $Z\in\mR$ and $\theta \not\equiv 0 \mod \pi$. As a consequence, we obtain after a straightforward calculation:
\begin{align}
-&\frac{\tilde{\zeta}(Z +i\theta)-\tilde{\zeta}(Z -i\theta)}{i \sin \theta}=\frac{1}{\cosh(Z)-\cos(\theta)}+ \notag \\
&2\sum\limits_{k=1}^\infty \frac{\cosh(Z) \cosh(kt)-\cos(\theta)}{\sinh^2(Z)+\cosh^2(kt)+\cos^2(\theta)-2\cosh(Z)\cosh(kt)\cos(\theta)}. \label{eq:second_derivative}
\end{align}
The right hand side is positive if $\theta \not\equiv 0 \mod \pi$. For this, we observe that \[\cos(\theta)<1\leq\cosh(Z)\leq\cosh(Z)\cosh(kt)\] and we use the estimate $\left(\cosh(Z)-\cosh(kt)\right)^2\geq 0$ to obtain
\begin{eqnarray*}
&\sinh^2(Z)+\cosh^2(kt)+\cos^2(\theta)-2\cosh(Z)\cosh(kt)\cos(\theta)\\
\geq& 2\cosh(Z)\cosh(kt)(1-\cos(\theta))+\cos^2(\theta)-1\\
=&(1-\cos(\theta))(2\cosh(Z)\cosh(kt)-\cos(\theta)-1)\\
>&0.
\end{eqnarray*}
It follows that $S$ is strictly concave if $\exp(i\theta_e) \in S^1_+$ for all edges $e\in E(\cG_W)$ and strictly convex if $\exp(i\theta_e) \in S^1_-$ for all edges $e$.

In contrast to our consideration in the previous paragraphs, these conditions are not necessary for (strict) concavity or convexity, i.e., $S$ might be strictly concave (or convex) also if there are labels $\theta_e$ on edges $e$ satisfying $\exp(i\theta_e) \in S^1_+$ as well as labels $\theta_{e^\prime}$ with $\exp(i\theta_{e^\prime}) \in S^1_-$, even in the integrable case.

For example, consider one white edge $e=(x,y)$ and assume that both $x$ and $y$ have exactly $4k+1$ other neighbors in $\cG_W$. In addition, let us suppose that $x$ and $y$ have no common neighbors. For simplicity, we now restrict to just this white graph. A corresponding quad-graph $\cG$ can be easily constructed.

We label the edge $e$ with $\theta_e=-\pi/2$, all other edges in the stars of the vertices $x$ and $y$ receive the label $\pi/2$. Since the labels around a vertex add up to $2k\pi$, the labeling is integrable. The actual reason why the conditions above are not necessary conditions for concavity or convexity is that the function $f:\mR \to \mR$,
\begin{equation*}
f(u):=i\left(\tilde{\zeta}(u+\frac{\pi}{2}i)-\tilde{\zeta}(u-\frac{\pi}{2}i)\right),
\end{equation*}
is periodic and positive. Therefore, there exist constants $K_+, K_->0$ such that $K_+\geq f(u)\geq K_-$ for all real $u$. Choosing the integer $k$ such that $(4k+1)K_->K_+$ yields a strictly concave action functional $S$.

\begin{remark}
Let us investigate the reality condition described by $X\in \mR$ if $x \in V_1$, $X \in \mR - \pi i$ if $x \in V_2$, and $\theta \in \mR \backslash \mZ\pi$, where $V_1 \dot{\cup} V_2$ is some partition of $V(\cG_W)$ into two vertex sets. Again, we define $X':=X$ if $x \in V_1$, $X':=X+\pi i$ if $x \in V_2$, $\theta'_e:=\theta_e$ if the vertices of the edge $e$ lie in the same vertex set, and $\theta'_e:=\theta_e +\pi$ if the vertices of $e$ are in different vertex sets. Then,
\begin{equation*}
\Phi(X,Y;\theta)=\frac
{\vartheta_1\left(\frac{X^\prime-Y^\prime+i\theta^\prime}{2i}\right)
\vartheta_1\left(\frac{X^\prime+Y^\prime+i\theta^\prime}{2i}\right)}
{\vartheta_1\left(\frac{X^\prime-Y^\prime-i\theta^\prime}{2i}\right)
\vartheta_1\left(\frac{X^\prime+Y^\prime-i\theta^\prime}{2i}\right)}.
\end{equation*}
Therefore, we obtain the same generalized discrete action functional as above if we replace $X,Y,\theta$ by $X^\prime,Y^\prime,\theta^\prime$.
\end{remark}

We come now to the reality condition $X,Y,\theta \in \mR i$, $\theta \notin ti\mZ$ and introduce new variables $X'$ and labels $\theta'$ by $X=iX'$, $\theta=i\theta'$. As before, we consider the variational formulation with respect to these new real variables.

Due to the $2ti$-periodicity of $\Phi$ in the $\theta$-variable, we may restrict $\theta \in (-t,t)i$, still assuming that $\theta\neq 0$. Since the zeroes of $\vartheta_1$ are given by $\pi\mZ +ti\mZ/2$, there is an open strip $U'$ containing the imaginary line, such that $(u\pm i\theta)/2i$ is not a zero of $\vartheta_1$ for all $u\in U'$. Now, we can argue exactly as before and get the generalized discrete action functional given by Formula~(\ref{eq:action}) and \begin{equation}\label{eq:L_Q4E}
iL(X',Y';\theta')=I_1(X-Y;\theta)+I_1(X+Y;\theta).
\end{equation}
Here, we adapt the definition of $I_1$ appropriately. The factor of $i$ is due to the different integration variable $X'=-iX$.

Now, the Hessian of $S$ is essentially the same as the one in Formula~(\ref{eq:D2_Q4}), just ignore the factor $-i$ and replace $(dX\pm dY)^2$ by $(dX'\pm dY')^2$. By a similar argument as before, we cannot expect necessary and sufficient conditions for (strict) concavity or convexity, but at least sufficient ones. To find such conditions, we investigate the real-valued function \[f(u;\theta):=\tilde{\zeta}(iu-\theta')-\tilde{\zeta}(iu+\theta')=2\re(\tilde{\zeta}(iu-\theta'))\] for real $u$. If $f$ is positive for all $u$ and suitable $\theta'$, $S$ is strictly convex, and if $f$ is negative, $S$ is strictly concave.

Since $\tilde{\zeta}$ is an odd function, it suffices to consider the case $\theta'>0$. Using that $\tilde{\zeta}$ is $2\omega$-periodic by construction, $f(\cdot;\theta')$ is $2\pi$-periodic. As a consequence, we can restrict to $u\in [0,2\pi]$.
 
First, we want to show that for fixed $u$, there exists a positive real number $\rho(u)$ such that $f(u;\theta')$ has the same sign for all $\theta'\in (0,\rho(u))$. Indeed, if $u$ is neither zero nor $2\pi$, $f(u;\cdot)$ is holomorphic around zero, nonconstant, and has a zero at zero. Since zeroes of holomorphic functions lie discrete in $\mC$, the existence for such a number $\rho$ follows. If $u=0$, $f(0;\cdot)$ has a pole at zero, but is holomorphic around zero, so the claim follows.

Noting that $f(\cdot;\cdot)$ is continuous in both variables (as long as we do not hit a pole), we can choose $\rho(u)$ depending continuously on $u$. Taking the minimum of $\rho(u)$ for all $u\in [0,2\pi]$, we get a positive number $r=r(t)$ such that $f(u;\theta')$ does not change its sign for all real $u$ and $\theta'\in (0,r(t))$. By investigating $f(0;\theta')$ one gets that $f(u;\theta')<0$ in these cases. Therefore, $S$ is strictly concave if $\theta'\in (0,r(t))$, i.e., if $\theta\in (0,r(t))i$.

\begin{remark}
We can deal with the reality condition $X\in \mR i$ if $x \in V_1$, $X \in \mR i - t$ if $x \in V_2$, and $\theta \in i(\mR\backslash \mZ t)$, where $V_1 \dot{\cup} V_2$ is some partition of $V(\cG_W)$ into two vertex sets, in the same way as we did for $X\in \mR$ if $x \in V_1$, $X \in \mR - \pi i$ if $x\in V_2$, noting that $\Phi$ is $2ti$-periodic in $\theta$.
\end{remark}

To conclude this subsection, we want to comment on rhombic period lattices, i.e., a half-period ratio $\tau$ of length one. In this setting, the original formulation of the quad-equation $(\textnormal{Q}4)$ given by the first author, Adler, and Suris in \cite{ABS03} seems to be easier to handle. Their three-leg function $\Phi$ is given by 
\begin{equation*}
 \Phi(X,Y;\theta)=\frac{\sigma(X-Y+i\theta) \sigma(X+Y+i\theta)}{\sigma(X-Y-i\theta) \sigma(X+Y-i\theta)}.
\end{equation*}
Here, $i\theta:=A-B$, where $A,B$ are certain values fulfilling $\wp(A)=\alpha, \wp(B)=\beta$.

Let $\omega$ and $\omega^\prime:=\tau \omega$ be half-periods of $\wp$ such that $\Omega:=\omega+\omega^\prime \in \mR^+$. Again, we start with the reality condition $X,Y,\theta_e \in \mR$ for all edges $e=(x,y)\in E(\cG_W)$, $\theta_e \notin \Omega\mZ$. The derivation of the generalized discrete action functional is similar to our consideration above, so we just give its Hessian
\begin{align*}
D^2S=&-i\sum\limits_{e=(x,y)\in E(\cG_W)} \left\{\zeta(X-Y+\theta_e)-\zeta(X-Y-\theta_e)\right\} (dX-dY)^2\\
&-i\sum\limits_{e=(x,y)\in E(\cG_W)}\left\{\zeta(X+Y+\theta_e)-\zeta(X+Y-\theta_e)\right\} (dX+dY)^2.
\end{align*}

Now, there exists $r=r(\tau)>0$ such that $S$ is strictly concave if $\theta_e\in (0,r)$ for all $e \in E(\cG_W)$ and strictly convex if $\theta_e\in (-r,0)$ for all edges $e$. The proof is almost identical to the one we gave in our investigation of strict concavity of the functional given by Formula~(\ref{eq:L_Q4E}) above. Just note that the periodicity of $\tilde{\zeta}$ is replaced by $\zeta(u+2\Omega)=\zeta(u)+2\zeta\left(\omega\right)+2\zeta\left(\omega^\prime\right)$.

In the very same way as we did for the reality conditions in the case of rectangular period lattices, we can consider the reality condition $X,Y,\theta_e \in \mR i$, and we may shift the variables for some of the vertices by $-\Omega$ or $\omega'-\omega$. The results are similar; however, we only get sufficient conditions for strict concavity for $\theta$ close to zero.


\subsection{Limit cases} \label{sec:limit}

To conclude this section, we want to demonstrate how the situation of $(\textnormal{Q}3)_{\delta=1}$ can be described as the limit of the setting in $(\textnormal{Q}4)$ when the (purely imaginary) half-period ratio $\tau$ goes to infinity.

As in Section~\ref{sec:Q4}, let $\tau=-t/(2\pi i)$. For simplicity, we just describe the reality condition $X,Y,\theta_e\in\mR$ for all edges $e$. We already noted in Sections~\ref{sec:Q31} and~\ref{sec:Q4} that the shifting $X \to X-\pi i$ for some vertices $x$ yields almost the same generalized discrete action functionals and we observed that the formulae for the reality condition $X,Y,\theta_e\in\mR i$ were almost identical (however, the conditions for convexity or concavity were different).

Since $\sn \rightarrow \sin$, the formulation of the quad-equation $(\textnormal{Q}4)$ given in the appendix converges to $(\textnormal{Q}3)_{\delta=1}$. The same is true for the three-leg functions $\Phi$ given in (\ref{eq:Phi_Q4}) and (\ref{eq:Phi_Q31}), observing that $\vartheta_1(v)/\vartheta_1(v')$ converges to $\sinh(iv)/\sinh(iv')$. Also, $\varphi$ defined in (\ref{eq:phi_Q4}) converges to $\varphi$ defined in (\ref{eq:phi_Q31}).

Equation~(\ref{eq:second_derivative}) shows that \[i\tilde{\zeta}(X\pm Y +i\theta)-i\tilde{\zeta}(X\pm Y -i\theta) \rightarrow \frac{ \sin(\theta)}{\cosh(X\pm Y)-\cos(\theta)}.\]
Thus, the Hessian of the generalized discrete action functional $S_{(\textnormal{Q}4)}$ given in (\ref{eq:D2_Q4}) converges to the Hessian of $S_{(\textnormal{Q}3)_{\delta=1}}$ given in (\ref{eq:D2_Q31}).

Now, the limit of $L_{(\textnormal{Q}4)}$ defined in (\ref{eq:L_Q4}) satisfies Conditions~(\ref{eq:symmetry}) and~(\ref{eq:derivative}) for the additive three-leg function of ${(\textnormal{Q}3)_{\delta=1}}$. But the same does $L_{(\textnormal{Q}3)_{\delta=1}}$ defined as the corresponding summand in $S_{(\textnormal{Q}3)_{\delta=1}}$ given in (\ref{eq:action_Q31}). In particular, their difference has to be constant in $X$ and $Y$. Evaluating at $X=Y=0$ we obtain
\begin{align*}
L_{(\textnormal{Q}4)}(0,0;i\theta)&=0,\\
L_{(\textnormal{Q}3)_{\delta=1}}(0,0;i\theta)&=4 \im \Li (\exp(i\theta)).
\end{align*}
As a consequence,
\begin{equation*}
S_{(\textnormal{Q}4)}\rightarrow S_{(\textnormal{Q}3)_{\delta=1}}-\sum\limits_{e=(x,y)\in E(\cG_W)}4 \im \Li (\exp(i\theta_e)).
\end{equation*}

In the end, we comment shortly on the more delicate limit from $(\textnormal{Q}3)_{\delta=1}$ to $(\textnormal{Q}3)_{\delta=0}$. If we let all variables $X$ go to $-\infty$, but preserve their differences, then \[\frac{\sinh\left(\frac{X+Y+i\theta}{2}\right)}{\sinh\left(\frac{X+Y-i\theta}{2}\right)}\rightarrow \exp(-i\theta),\] so the three-leg function of $(\textnormal{Q}3)_{\delta=1}$ converges to the one of $(\textnormal{Q}3)_{\delta=0}$. Also, the Hessian of $S_{(\textnormal{Q}3)_{\delta=1}}$ converges to the Hessian of $S_{(\textnormal{Q}3)_{\delta=0}}$ as becomes immediate from Formulae~(\ref{eq:D2_Q31}) and~(\ref{eq:D2_Q30}). Having the relation of the generalized discrete action functionals to the circle pattern functionals of the first author and Springborn in \cite{BSp04} in mind, $X\to -\infty$ means that the radii of the hyperbolic circle pattern go to zero. This corresponds to seeing the Euclidean plane as a limit of hyperbolic planes with curvature going to 0.

Note that for the reality condition $X,Y,\theta_e\in\mR i$ for all $e\in E(\cG_W)$, we do not such have a limit since \[\frac{\sin\left(\frac{u+i\theta}{2}\right)}{\sin\left(\frac{u-i\theta}{2}\right)}\] does not converge if $u\to\pm \infty$.


\section[Existence and uniqueness of solutions of (Q3)- and (Q4)-DBVP]{Existence and uniqueness of solutions of (Q3)- and (Q4)-Dirichlet boundary value problems}\label{sec:Q_Dirichlet}

This section is devoted to the study of Dirichlet boundary value problems for $(\textnormal{Q}3)$ and $(\textnormal{Q}4)$, considering the reality condition $X,Y,\theta_e \in \mR$ for all edges $e=(x,y)$ and restricting to one of the cases $\exp(i\theta_e)\in S^1_+$ or $\exp(i\theta_e)\in S^1_-$. The reality conditions where some of the variables are shifted to $X-\pi i$ can be handled the same way.

The reason why we not consider the same problems for $(\textnormal{Q}1)$ and $(\textnormal{Q}2)$ or for the other reality conditions of $(\textnormal{Q}3)$ and $(\textnormal{Q}4)$ is that the conditions for convexity or concavity given in Theorem~\ref{th:main} cannot be achieved by an integrable labeling as we will see in Proposition~\ref{prop:incompatibility} of Section~\ref{sec:integrable_cases}. In contrast, the conditions for $(\textnormal{Q}3)$ and $(\textnormal{Q}4)$ above might be realized by an integrable labeling. Still, we do not require the labeling to be integrable in the following.

Note that Theorem~\ref{th:Q3} below is due to the first author and Springborn \cite{BSp04}, we just exploit the identification of the generalized discrete action functionals of $(\textnormal{Q}3)$ with the corresponding circle pattern functionals as explained in Sections~\ref{sec:Q30} and~\ref{sec:Q31}.

\begin{theorem}\label{th:Q3}
Let $\theta$ be a labeling of $E(\cG_W)$ such that $0<\theta_e<\pi$ for all edges $e$. Moreover, let $\Theta_x \in \mR^+$ be given. Then, the corresponding generalized discrete action functionals of $(\textnormal{Q}3)$ and $(\textnormal{Q}4)$ given by Formulae~(\ref{eq:action_Q30}) and~(\ref{eq:action_Q31}) have an extremum on the subspace $U$ defined in (\ref{eq:U}) or $\mR^{|V(\cG_W)|}$, respectively, if and only if the following two conditions are satisfied:
\begin{enumerate}
\item $\sum\limits_{e \in E(\cG_W)} 2\theta^*_e-\sum\limits_{x \in V(\cG_W)} \Theta_x$ is equal to zero in the case of $(\textnormal{Q}3)_{\delta=0}$ and greater than zero in the case of $(\textnormal{Q}3)_{\delta=1}$;
\item $\sum\limits_{e \in E^\prime} 2\theta^*_e-\sum\limits_{x \in V^\prime} \Theta_x>0$ for all nonempty $V^\prime \varsubsetneq V(\cG_W)$, where $E^\prime$ denotes the set of all edges incident to some vertex in $V^\prime$.
\end{enumerate}
The extremum is unique if it exists.
\end{theorem}

\begin{proof}
In Sections~\ref{sec:Q30} and~\ref{sec:Q31} we have seen that one can identify the generalized discrete action functionals in the case of $(\textnormal{Q}3)$ with the circle pattern functionals of \cite{BSp04} under certain conditions (e.g., $0<\theta<\pi$ and $\Theta>0$). Thus, we can adapt the first author's and Springborn's result in \cite{BSp04} to our setting. They gave a complete answer to the question when extrema of these functionals exist. Their proof consisted of two steps. First, they showed that existence of an extremum is equivalent to the existence of so-called coherent angle systems; second, they used the feasible flow theorem of network theory to prove that coherent angle systems exist if and only if the conditions of the theorem hold. We refer the reader to their paper \cite{BSp04} for more details.

Note that the uniqueness of the extremum follows from strict concavity of the functionals.
\end{proof}

\begin{remark}
Using Equations~(\ref{eq:remark_Q30}) and~(\ref{eq:remark_Q31}), we get an analogous theorem in the case of $\pi<\tilde{\theta}_e<2\pi$, replacing $\theta_e=\pi -s_e$ by $\tilde{\theta}_e=\pi +s_e$ and choosing $\Theta_x=-\tilde{\Theta}_x$.
\end{remark}

\begin{definition}
We consider the generalized discrete action functional $S$ corresponding to the reality condition $X,Y,\theta_e \in \mR$ in the case of $(\textnormal{Q}3)$ or $(\textnormal{Q}4)$. A \textit{Dirichlet boundary value problem} asks for the existence of a critical point of $S$ under the constraint that the variables corresponding to some preassigned vertices of $\cG_W$ are fixed.
\end{definition}

Note that in the definition above, we do not require that the preassigned vertices lie on the boundary of $\cG_W$. In the nongeometric case, the situation simplifies compared to Theorem~\ref{th:Q3} if all $\Theta$ vanish. The proof follows the lines of the first author's and Springborn's proof of Theorem~\ref{th:Q3}. But now, coherent angle systems are not necessary any more, since all $\Theta$ are zero.

\begin{theorem}\label{th:Dirichlet_Q31}
Let $\theta$ be a labeling of $E(\cG_W)$ such that $\exp(i\theta_e) \in S^1_+$ for all edges $e$ or $\exp(i\theta_e) \in S^1_-$ for all edges $e$. Then, the Dirichlet boundary value problem corresponding to $(\textnormal{Q}3)_{\delta=1}$ is uniquely solvable if $\Theta_x=0$ for all vertices $x$ corresponding to a nonfixed variable.
\end{theorem}
\begin{proof}
The uniqueness follows from strict concavity or convexity of the generalized discrete action functional $S$ due to Theorem~\ref{th:main}. By Equation~(\ref{eq:remark_Q31}), we may restrict to $0<\theta_e<\pi$ for all edges $e$, such that $S$ is strictly concave. Then, for some constant $C$ given by the contributions of $\Theta_x X$ for fixed variables $X$, $S+C$ is equal to
\begin{align*}
&-\sum\limits_{e=(x,y)\in E(\cG_W)} \left\{\im \Li(\exp(X-Y+i \theta_e))+\im \Li(\exp(Y-X+i \theta_e))\right\}\\
&-\sum\limits_{e=(x,y)\in E(\cG_W)}\left\{\im \Li(\exp(X+Y+i \theta_e))+\im \Li(\exp(-X-Y+i \theta_e))\right\}.
\end{align*}

Now, for any $u \in \mR$ and $0<\theta<\pi$, we have
\begin{equation*}
\im \Li(\exp(u+i \theta))+\im \Li(\exp(-u+i \theta))>(\pi - \theta) \left|u\right|.
\end{equation*}

As a consequence, we obtain
\begin{align*}
S&<-\sum\limits_{e=(x,y)\in E(\cG_W)} (\pi - \theta_e) (\left|X-Y\right|+\left|X+Y\right|)-C\\
&\leq -\min_{e\in E(\cG_W)} (\pi-\theta_e) \sum\limits_{e=(x,y)\in E(\cG_W)} 2\max(\left|X\right|,\left|Y\right|)-C.
\end{align*}
In particular, $S  \to -\infty$ if the absolute value of some variable $X$ goes to infinity. Thus, a maximum of the strictly concave functional $X$ exists.
\end{proof}

\begin{remark}
In Section~\ref{sec:Q30}, we have seen that $\Theta_x>0$ (or $\Theta_x<0$) for all vertices $x \in V(\cG_W)$ is necessary to obtain solutions of the discrete Laplace-type Equation~(\ref{eq:Laplace}) corresponding to $(\textnormal{Q}3)_{\delta=0}$ if $\exp(i\theta_e) \in S^1_+$ (or $\exp(i\theta_e) \in S^1_-$) for all edges $e$. On that account, Theorem~\ref{th:Dirichlet_Q31} would not be reasonable in the case of $(\textnormal{Q}3)_{\delta=0}$.
\end{remark}

In the case of the most general quad-equation $(\textnormal{Q}4)$, it turns out that the corresponding Dirichlet boundary value problem is always solvable.

\begin{theorem}\label{th:Dirichlet_Q4}
Let $\theta$ be a labeling of $E(\cG_W)$ such that $\exp(i\theta_e) \in S^1_+$ for all white edges $e\in E(\cG_W)$ or $\exp(i\theta_e) \in S^1_-$ for all edges $e$. Then, the Dirichlet boundary value problem corresponding to $(\textnormal{Q}4)$ is uniquely solvable.
\end{theorem}
\begin{proof}
By Theorem~\ref{th:main}, the generalized discrete action functional $S$ given by Formula~(\ref{eq:L_Q4}) is strictly concave or strictly convex. Hence, a extremum is unique if it exists. In the following, we will restrict to the case that $S$ is strictly concave, the other case can be dealt with in exactly the same way.

As in Section~\ref{sec:Q4}, let $\tau=-t/(2\pi i)$ denote the half-period ratio of the Jacobi theta function $\vartheta_1$. Then,
\[\frac{\partial \varphi}{\partial X}(X,Y;\theta)=-i(\tilde{\zeta}(X+Y+i \theta)-\tilde{\zeta}(X+Y-i \theta))-i(\tilde{\zeta}(X-Y+i \theta)-\tilde{\zeta}(X-Y-i \theta))\]
is a smooth and negative function that is $2t$-periodic in both variables $X$ and $Y$. $\partial \varphi/\partial Y$ is also a smooth and double-periodic function. It follows that both derivatives of $\varphi$ are bounded, so $\varphi(\cdot,\cdot;\theta)$ is Lipschitz-continuous. Since the set of edges is finite, we may choose the Lipschitz constant $L$ uniformly for all $\theta_e$.

A straightforward calculation shows that for all $\theta \in (0,\pi]$,
\begin{align}
\Phi(X+2t,Y;\theta)&=\Phi(X,Y;\theta) \exp(8i\theta), \label{eq:Phi_X}\\
\Phi(X,Y+2t;\theta)&=\Phi(X,Y;\theta). \label{eq:Phi_Y}
\end{align}
Let us choose $\mu=\mu(X,Y;\theta), \eta=\eta(X,Y;\theta)$ in such a way that
\begin{align*}
\varphi(X+2t,Y;\theta)-\varphi(X,Y;\theta)&=\mu(X,Y;\theta),\\
\varphi(X,Y+2t;\theta)-\varphi(X,Y;\theta)&=\eta(X,Y;\theta).
\end{align*}
By continuity and Equations~(\ref{eq:Phi_X}) and~(\ref{eq:Phi_Y}), $\eta$ is constant and $\mu=\mu(\theta)$ depends continuously only on $\theta \in (0,\pi]$. Now, $\partial \varphi/\partial X<0$, so $\mu<0$; and $\Phi(X,Y;\pi) \equiv 1$, so $\eta=0$.

Let us suppose that the Dirichlet boundary value problem is not solvable. Then, there exists a sequence \[(\{X_j(n)\}_{j=1}^v)_{n=0}^\infty\] of variables (some of them fixed), $v=\left|V(\cG)\right|$, such that at least one variable is unbounded and such that the sequence \[(S(n))_{n=0}^\infty, \quad S(n):=S(X_1(n),\ldots,X_v(n)),\] is strictly increasing. In the following, we will subsequently pass to subsequences with certain properties and for the ease of notation, we will use the same indexing for the subsequences as for the original sequence.

We enumerate the $X$-variables in such a way that the $X_1$-variable is unbounded. By passing to a subsequence, we can achieve that the sequence of $X_1(n)$ is strictly increasing or decreasing. Similarly, we can assume for any $j$ that $X_j(n)$ either converges to a constant $a_j$ or strictly monotonously to $\pm \infty$. Note that fixed variables $X_j$ are just constantly $a_j$.

For a variable $X_j$ such that $X_j(n) \rightarrow \pm \infty$, consider the induced sequence on the compact circle $\mR/2t\mZ$. Then, we can choose a convergent subsequence, i.e., there exist $a_j \in \mR$ and a sequence $(m_j(n))_{n=0}^\infty$ of integers such that \[X_j(n)-2m_j(n)t \rightarrow a_j.\] We may suppose that $m_j(n) \rightarrow \pm \infty$ strictly monotonously. Using the same notation, we can include the case that $X_j(n)$ converges to a constant by setting $m_j(n)=0$ for all $n$.

Now, let $0<\varepsilon<2t$ be fixed. By passing to a subsequence, we can assume that 
\begin{equation*}
\left| X_j(n)-2m_j(n)t-a_j \right| < \varepsilon
\end{equation*}
for all $j$ and $n$. By our consideration above, we have for any edge $e=(x_j,x_k)$:
\begin{align*}
\varphi(X_j(n),X_k(n^\prime);\theta_e)&=\varphi(a_j+2m_j(n)t,a_k+2m_k(n^\prime)t;\theta_e)+B\\
&=\varphi(a_j,a_k;\theta_e)+\mu(\theta_e)m_j(n)+B
\end{align*}
for some real number $B$ depending on $e,j,k,n,n'$ satisfying $\left|B\right|<L\sqrt{2\varepsilon^2}$. More general, we see for any real $X=a+2pt+q$, where $p \in \mZ$ and $\left|q\right|<2t$, that
\begin{equation*}
\varphi(X,X_k(n^\prime);\theta_e)=\varphi(a,a_k;\theta_e)+\mu(\theta_e)p+B,
\end{equation*}
where $\left|B\right|<L\sqrt{\varepsilon^2+4t^2}$.

For a positive real number $M$ that we will fix later, we define
\begin{equation*}
p_{\text{max}}:=\frac{1}{\min\limits_{e \in E(\cG_W)}|\mu(\theta_e)|}\left(L\sqrt{\varepsilon^2+4t^2}+M+\max_{X_j(n) \text{ diverges}} \max_{e=(x_j,x_k)} \left|\varphi(a_j,a_k;\theta_e)\right|\right).
\end{equation*}
By the results we obtained so far, we get for an edge $e=(x_j,x_k)$ such that $X_j(n)$ diverges and $X=a+2pt+q$ as above with $\left|p\right|>p_{\text{max}}$:
\begin{equation*}
\delta\varphi(X,X_k(n^\prime);\theta_e) < -M,
\end{equation*}
where $\delta:=\textnormal{sgn}(p)$. By passing to a subsequence, we suppose that $\left|m_j(n)\right|>p_{\text{max}}$ for all $j$ such that $X_j(n)$ diverges. Now, define
\begin{equation*}
C:=\sqrt{2}L\varepsilon+\max_{X_j(n) \text{ converges}} \max_{e=(x_j,x_k)} \left|\varphi(a_j,a_k;\theta_e)\right|.
\end{equation*}
Then, $\left|\varphi(X,X_k(n^\prime);\theta_e)\right|< C$ for any edge $e=(x_j,x_k)$ such that $X_j(n)$ converges, provided that $\left|X-a_j\right|<\varepsilon$. Finally, define \[E:=\max\limits_{x \in V(\cG_W)} \left|\Theta_x\right|, \tilde{C}:=C+E, \tilde{M}:=M-E.\] We will choose $M$ later in such a way that $\tilde{M}>0$.

As a short summary of our consideration above, we have shown that
\begin{align*}
&\delta\frac{\partial S}{\partial X_j} (X_1(n+1),\ldots,X_{j-1}(n+1),X_j,X_{j+1}(n),\ldots,X_v(n))\\
=&\delta\left(\sum_{e=(x_j,x_k) \textnormal{ edge}}\varphi(X_j,X_k(n \text{ or } n+1);\theta_e)- \Theta_{x_j}\right)\\
<& -\tilde{M}
\end{align*}
if $X_j \in \text{conv}\{X_j(n),X_j(n+1)\}$ and $X_j(n) \rightarrow \delta \infty$, $\delta=\pm 1$; and
\begin{align*}
&\left|\frac{\partial S}{\partial X_j}(X_1(n+1),\ldots,X_{j-1}(n+1),X_j,X_{j+1}(n),\ldots,X_v(n))\right|\\
=&\left|\sum_{e=(x_j,x_k) \textnormal{ edge}}\varphi(X_j,X_k(n \text{ or } n+1);\theta_e)+ \Theta_{x_j}\right|\\
< &\tilde{C}v
\end{align*}
if $X \in \text{conv}\{X_j(n),X_j(n+1)\}$ and $X_j(n) \rightarrow a_j$.

Without loss of generality, we enumerate the vertices in such a way that we have $X_j(n) \rightarrow \pm \infty$ if and only if $j \leq k$. If we integrate along the piecewise straight path \begin{align*}
(X_1(n),\ldots,X_v(n))&-(X_1(n+1),X_2(n),\ldots,X_v(n))\\
&-(X_1(n+1),X_2(n+1),\ldots,X_v(n))\\
&-\ldots\\
&-(X_1(n+1),\ldots,X_v(n+1)),
\end{align*} we obtain
\begin{align*}
S_{n+1}-S_n&<-\tilde{M} \sum_{j=1}^k \left|X_j(n+1)-X_j(n)\right|+\tilde{C}v\sum_{j=k+1}^v \left|X_j(n+1)-X_j(n)\right|\\
&<-\tilde{M}\left|X_1(n+1)-X_1(n)\right|+\tilde{C}v(v-k)2\varepsilon.
\end{align*}

Again, we can achieve $\left|X_1(n+1)-X_1(n)\right|\geq 1$ for all $n$ by passing to a subsequence. After all, we choose $M$ large enough such that $M>E+2\tilde{C}v(v-k)\varepsilon$. But then, $S_{n+1}-S_n<0$, contradicting our assumption that $S(n)$ is increasing. Consequently, $S$ has a (unique) maximum.
\end{proof}

\begin{remark}
Although the concavity condition we gave for rhombic lattices in the end of Section~\ref{sec:Q4} is not integrable, we would like to mention that Theorem~\ref{th:Dirichlet_Q4} also holds for unitary half-periods $\tau$. For this, we use that
\begin{equation*}
\sigma(u+2\Omega)=-\exp(2\eta^\prime(u+\Omega))\sigma(u),
\end{equation*}
where $\omega$ and $\omega^\prime:=\tau \omega$ are half-periods of $\wp$ such that $\Omega:=\omega+\omega^\prime \in \mR^+$ and $\eta^\prime=\zeta\left(\omega\right)+\zeta\left(\omega^\prime\right)$. This gives results analogous to Equations~(\ref{eq:Phi_X}) and~(\ref{eq:Phi_Y}), such that the same arguments work in the rhombic case as well.
\end{remark}


\section{Integrable cases}\label{sec:integrable_cases}

In this section, we want to discuss when a labeling of $E(\cG)$ fulfilling the conditions for concavity or convexity in Theorem~\ref{th:main} is integrable. First, we observe that the conditions $\pm\theta_e>0$ and $\pm i\theta_e>0$ are never integrable:

\begin{figure}[htbp]
\begin{center}
\beginpgfgraphicnamed{induced-labeling}
\begin{tikzpicture}
[black/.style={circle,draw=black,fill=white,thin,inner sep=0pt,minimum size=1.2mm},
white/.style={circle,draw=black,fill=black,thin,inner sep=0pt,minimum size=1.2mm}]
\node[black] (x) [label=below left:$x$] at (0,0) {};
\node[white] (u)  at (2,0) {};
\node[black] (y) [label=above right:$y_k$] at (2,2) {};
\node[white] (v)  at (0,2) {};
\draw[dashed] (x) -- node[midway,right] {$\theta_e$} (y);
\draw (x) -- node[midway,below] {$\alpha_k$} (u) -- node[midway,right] {$\alpha_{k+1}$} (y)-- node[midway,above] {$\alpha_k$} (v)-- node[midway,left] {$\alpha_{k+1}$} (x);
\end{tikzpicture}
\endpgfgraphicnamed
\caption{Induced labeling}
\label{fig:ind_labeling}
\end{center}
\end{figure}
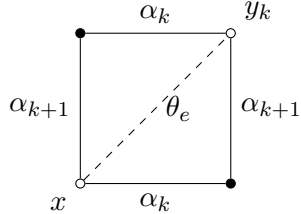

\begin{proposition}\label{prop:incompatibility}
If the quad-graph $\cG$ contains at least one inner white vertex, there is no integrable labeling of $E(\cG_W)$ such that $\delta\theta_e>0$ for all edges $e$ of $\cG_W$ and some fixed nonzero $\delta \in \mC$.
\end{proposition}
\begin{proof}
If we consider all edges $e=(x,y_k)$ incident to an inner white vertex $x$ and $\theta_e=\alpha_k-\alpha_{k+1}$ is induced from a labeling of $E(\cG)$ as shown in Figure~\ref{fig:ind_labeling}, then
\begin{equation}\label{eq:zerosum}
 \sum\limits_{(x,y_k) \in E(\cG_W)} \theta_e = \sum\limits_{(x,y_k) \in E(\cG_W)} (\alpha_k-\alpha_{k+1}) = 0.
 \end{equation}
Hence, not all $\delta\theta$ can be positive.
\end{proof}

Therefore, only the conditions $\exp(i\theta_e) \in S^1_+$ for all edges $e$ or $\exp(i\theta_e) \in S^1_-$ for all $e\in E(\cG_W)$ can be integrable. Since the labeling $-\theta$ is integrable if and only if the labeling $\theta$ is integrable, we can restrict to the case  $\exp(i\theta) \in S^1_+$, i.e., $\theta \equiv (0,\pi) \mod 2\pi$. To have a short notation, we say that such a labeling is \textit{convex}.

In the following, we first determine all possible convex integrable labelings on some very special graphs and then, we will discuss convex integrable labelings on rhombic-embeddable quad-graphs.

\begin{figure}[htbp]
\begin{center}
\beginpgfgraphicnamed{Z2}
\begin{tikzpicture}
[white/.style={circle,draw=black,fill=white,thin,inner sep=0pt,minimum size=1.2mm},
black/.style={circle,draw=black,fill=black,thin,inner sep=0pt,minimum size=1.2mm}]
\foreach \x in {-3,...,3}
	\foreach \y in {-3,...,3}
		{\pgfmathparse{int(add(\x,\y))};
		\ifthenelse{\isodd{\pgfmathresult}}{\node[black] (p_\x_\y)  at (\x,\y) {};}{\node[white] (p_\x_\y)  at (\x,\y) {};}
		}
\foreach \x in {-3,...,3}
	\foreach \y in {-3,...,2}
		{\pgfmathparse{int(add(\y,1))};
		\ifthenelse{\equal{\x}{3}}{\draw (p_\x_\y) --node[midway,right]{$\beta_{\y}$} (p_\x_\pgfmathresult);}{\draw (p_\x_\y) -- (p_\x_\pgfmathresult);}
		}
\foreach \x in {-3,...,2}
	\foreach \y in {-3,...,3}
		{\pgfmathparse{int(add(\x,1))};
		\ifthenelse{\equal{\y}{-3}}{\draw (p_\x_\y) --node[midway,below]{$\alpha_{\x}$} (p_\pgfmathresult_\y);}{\draw (p_\x_\y) --  (p_\pgfmathresult_\y);}
		}
\foreach \x in {-3,...,3}
	\foreach \y in {-3,...,3}
		{
		\draw[dashed] (p_3_\y) -- +(1,0);
		\draw[dashed] (p_-3_\y) -- +(-1,0);
		\draw[dashed] (p_\x_3) -- +(0,1);
		\draw[dashed] (p_\x_-3) -- +(0,-1);
		}
\end{tikzpicture}
\endpgfgraphicnamed
\caption{$\mZ^2$}
\label{fig:Z2}
\end{center}
\end{figure}

Given a finite sub-quad-graph of $\mZ^2$, we consider only labelings of its edges that are actually induced from a labeling of all edges $\mZ^2$, see Figure~\ref{fig:Z2}. It is easy to check that all such labelings that induce a convex labeling are given by the following (up to multiples of $2\pi$ and adding a fixed number to all labels):
\begin{align*}
\beta_0 &=0,\\
\alpha_{2j} & \in (0,\pi),\\
\alpha_{2j+1} & \in (\pi,2\pi),\\
\beta_{2j+1} & \in (\max\limits_{k \in M_j} \alpha_{2k},\min\limits_{k \in M_j} \alpha_{2k}+\pi) \cap (\max\limits_{k \in N_j} \alpha_{2k+1}-\pi,\min\limits_{k \in N_j} \alpha_{2k+1}),\\
\beta_{2j} & \in \{(\max\limits_{k \in M^\prime_j} \alpha_{2k} +\pi,2 \pi) \cup [0,\min\limits_{k \in M^\prime_j} \alpha_{2k})\} \cap \{(\max\limits_{k \in N^\prime_j} \alpha_{2k+1},2 \pi) \cup [0,\min\limits_{k \in N^\prime_j} \alpha_{2k+1}-\pi)\}
\end{align*}
for all integers $j$. Here, $M_j,M^\prime_j,N_j,N^\prime_j$ denote the sets of all integers $k$ such that there exists a quadrilateral in $\mZ^2$ with labels $\alpha_{2k},\beta_{2j+1}$; $\alpha_{2k},\beta_{2j}$; $\alpha_{2k+1},\beta_{2j+1}$; $\alpha_{2k+1},\beta_{2j}$, respectively.

Note that regardless which $\alpha_j$ we choose above, $\beta_{2j}=0$ and $\beta_{2j+1}=\pi$ for all $j$ will always satisfy the last two conditions. The result can be easily extended to infinite sub-quad-graphs of $\mZ^2$, replacing minima and maxima by infima and suprema.

\begin{definition}
A \textit{spider-graph} is the (infinite) quad-graph that is constructed in the following way: Take infinitely many concentric regular $2n$-gons, $n \geq 2$, that are equally spaced. Then, add the radial edges between two successive polygons and divide the central polygon into quadrilaterals by adding $n-2$ parallel diagonals.
\end{definition}

An example of a spider-graph is given in Figure~\ref{fig:spider8} ($n=4$).

\begin{figure}[htbp]
\begin{center}
\beginpgfgraphicnamed{spider8}
\begin{tikzpicture}
[white/.style={circle,draw=black,fill=white,thin,inner sep=0pt,minimum size=1.2mm},
black/.style={circle,draw=black,fill=black,thin,inner sep=0pt,minimum size=1.2mm}]
\foreach \r in {1,...,4}
	\foreach \a in {22,67,...,337}
		{\pgfmathparse{int(\r+int(\a))};
		\ifthenelse{\isodd{\pgfmathresult}}{\node[black] (p_\r_\a)  at (canvas polar cs:angle=\a+0.5,radius=\r cm) {};}{\node[white] (p_\r_\a)  at (canvas polar cs:angle=\a+0.5,radius=\r cm) {};}
		}
\foreach \r in {1,...,4}
		{
		\draw (p_\r_22) -- node[midway,above right]{$\alpha_2$} (p_\r_67);
		\draw (p_\r_67) -- node[midway,above]{$\alpha_1$} (p_\r_112);
		\draw (p_\r_112) -- node[midway,above left]{$\alpha_2$} (p_\r_157);
		\draw (p_\r_157) -- node[midway,left]{$\alpha_3$} (p_\r_202);
		\draw (p_\r_202) -- node[midway,below left]{$\alpha_4$} (p_\r_247);
		\draw (p_\r_247) -- node[midway,below]{$\alpha_1$} (p_\r_292);
		\draw (p_\r_292) -- node[midway,below right]{$\alpha_4$} (p_\r_337);
		\draw (p_\r_337) -- node[midway,right]{$\alpha_3$} (p_\r_22);
		}
\foreach \r in {1,...,3}
	\foreach \a in {22,67,...,337}
	{\pgfmathparse{int(\r+1)}
	\draw (p_\r_\a) -- node[midway]{$\beta_\r$} (p_\pgfmathresult_\a);
	}
\foreach \a in {22,67,...,337}
	{\draw[dashed] (p_4_\a)--(canvas polar cs:angle=\a+0.5,radius=5cm);
	}
\draw (p_1_22) -- node[midway,above]{$\alpha_1$} (p_1_157);
\draw (p_1_202) -- node[midway,below]{$\alpha_1$} (p_1_337);
\end{tikzpicture}
\endpgfgraphicnamed
\caption{Spider-graph with central octagon}
\label{fig:spider8}
\end{center}
\end{figure}
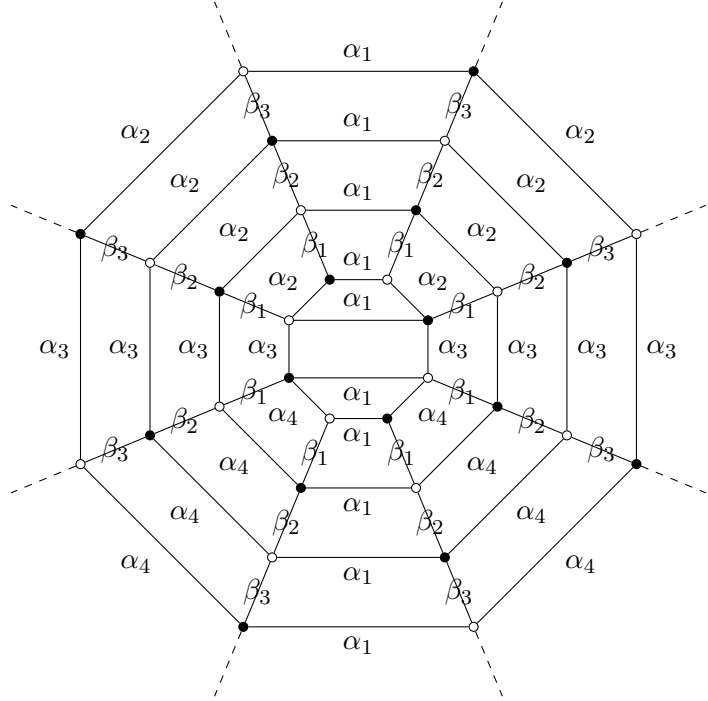

All labelings of the edges of a spider-graph that induce a convex labeling are given by the following (up to multiples of $2\pi$ and adding a fixed number to all labels):
\begin{align*}
\alpha_1 &=0,\\
\alpha_{2j+1} & \in (0,\pi),\\
\alpha_{2j} & \in (\pi,2\pi),\\
\beta_{2j-1} & \in (\max\{\max\limits_{k} \alpha_{2k}-\pi,\max\limits_{k} \alpha_{2k+1}\},\pi),\\
\beta_{2j} & \in (\max\{\max\limits_{k} \alpha_{2k},\max\limits_{k} \alpha_{2k+1}+\pi\},2\pi)
\end{align*}
for all nonnegative integers $j$.

\begin{definition}
A \textit{strip} in a planar quad-graph $\cG$ is a path on its dual such that two successive faces share an edge and the strip leaves a face in the opposite edge where it enters it. Moreover, strips are assumed to have maximal length, i.e., there are no strips containing it apart from itself.
\end{definition}

Rhombic planar quad-graphs are characterized by the following proposition of Kenyon and Schlenker \cite{KSch05}:

\begin{proposition}\label{prop:rhombic}
A planar quad-graph $\Lambda$ admits a combinatorially equivalent embedding in $\mC$ with all rhombic faces if and only if the following two conditions are satisfied:
\begin{itemize}
\item No strip crosses itself or is periodic.
\item Two distinct strips cross each other at most once.
\end{itemize}
\end{proposition}

As a consequence, spider-graphs possess no rhombic embedding into the complex plane. But rhombic-embeddable graphs are an interesting and important subclass of quad-graphs. For example, rhombic quad-graphs played an important role in the work \cite{BoMeSu05} of the first author, Mercat, and Suris on linear and nonlinear theories of discrete complex analysis.

Before we come to rhombic-embeddable quad-graphs, we will first prove a characterization of integrable labelings on \textit{simply-connected} quad-graphs. These are quad-graphs where the part of the surface that $F(\cG)$ covers is simply-connected. For this, we first note that any labeling of $E(\cG)$ induces not only a labeling on the edges of the white graph $\cG_W$, but also on $E(\cG_B)$. Indeed, the induced label $\theta_{e^*}$ on an edge $e^* \in E(\cG_B)$ is just defined as $\theta_{e^*}=-\theta_e$, where $\theta_e$ is the label induced on the dual edge $e \in E(\cG_W)$. Since opposite edges of any quadrilateral carry the same label, this definition is consistent with our notion of an induced labeling.

The following three propositions are inspired by the paper \cite{BoMeSu05} of the first author, Mercat, and Suris.

\begin{proposition}\label{prop:integrability}
Let $\cG$ be simply-connected. Suppose that a labeling of $E(\cG_W)$ with real numbers is given, where the labels $\theta$ are considered the same if they differ by a multiple of $2\pi$ only.

Then, this labeling is integrable if and only if the following two equations are satisfied for all inner white vertices $x$ and inner black vertices $x^*$:
\begin{align}
 \prod\limits_{e=(x,y_k)\in E(\cG_W)} \exp(i\theta_e) =1, \label{eq:vertex}\\
 \prod\limits_{e^*=(x^*,y_k^*)\in E(\cG_B)}\exp(i\theta_{e^*}) =1. \label{eq:face}
\end{align}
\end{proposition}
\begin{proof}
The forward implication is given by the multiplicative analog of (\ref{eq:zerosum}). For the backward implication, we notice that (\ref{eq:vertex}) allows us to integrate the labeling on the edges of the star of any white vertex. Now, we just have to check that this consistently defines a labeling on $E(\cG)$. This is the case if and only if the sum of $\theta_e$ for all white edges $e$ of a given cycle vanishes modulo $2\pi$. Since $\cG$ is simply-connected, this is true if it is true for any elementary cycle corresponding to the star of a black vertex. Using $\theta_{e^*}=-\theta_e$, we can replace the white edges by its dual black edges. Then, the sum of all $\theta_{e^*}$ around the black vertex vanishes if and only if Equation~(\ref{eq:face}) is satisfied.
\end{proof}

\begin{proposition}\label{prop:rhombic1}
If the white graph $\cG_W$ or the black graph $\cG_B$ is bipartite, then any (possibly ramified) rhombic embedding of $\cG$ into $\mC$ yields a labeling of edges of $\cG$ such that the induced labeling on $E(\cG_W)$ is convex. If $\cG_W$ is bipartite, the induced labels are given by the angles $\theta$ of the rhombi at the white vertices; if $\cG_B$ is bipartite, the induced labels are given by the angles $\theta^*=\pi - \theta$ of the rhombi at the black vertices.
\end{proposition}
\begin{proof}
Since $\cG_W$ (or $\cG_B$) is bipartite, we can choose a partition into two types of white (or black) vertices, say of type 1 and of type 2. We orient all edges of $\cG$ in such a way that they always start in a white (or black) and end in a black (or white) vertex, see Figure~\ref{fig:rhombic_bipartite}.

\begin{figure}[htbp]
   \centering
		\subfloat[$\cG_W$ bipartite]{
    \beginpgfgraphicnamed{white-bipartite}
			\begin{tikzpicture}
			[white1/.style={circle,draw=black,fill=white,thin,inner sep=0pt,minimum size=1.2mm},white2/.style={rectangle,draw=black,fill=white,thin,inner sep=0pt,minimum size=1.2mm},black/.style={circle,draw=black,fill=black,thin,inner sep=0pt,minimum size=1.2mm}]
		\node[white1] (w1) [label=above right:$\theta^*$]  at (-1,-1) {};
		\node[white2] (w2) [label=below left:$\theta^*$] at (1,1) {};
		\node[black] (b1) [label=above left:$\theta$] at (1,-1) {};
		\node[black] (b2) [label=below right:$\theta$] at (-1,1) {};
		\draw[decoration={markings,mark=at position 1 with {\arrow[scale=2]{>}}},
    postaction={decorate}, shorten >=0.4pt] (w1) -- node[midway,below] {$\exp(i\alpha)$} (b1);
		\draw[decoration={markings,mark=at position 1 with {\arrow[scale=2]{>}}},
    postaction={decorate}, shorten >=0.4pt] (w2) -- node[midway,right] {$-\exp(i\beta)$} (b1);
		\draw[decoration={markings,mark=at position 1 with {\arrow[scale=2]{>}}},
    postaction={decorate}, shorten >=0.4pt] (w1) -- node[midway,left] {$\exp(i\beta)$} (b2);
		\draw[decoration={markings,mark=at position 1 with {\arrow[scale=2]{>}}},
    postaction={decorate}, shorten >=0.4pt] (w2) -- node[midway,above] {$-\exp(i\alpha)$} (b2);
		\draw (-0.2,-1) arc (0:90:8mm);
		\draw (1,-0.2) arc (90:180:8mm);
		\draw (0.2,1) arc (180:270:8mm);
		\draw (-1,0.2) arc (270:360:8mm);
		\end{tikzpicture}
		\endpgfgraphicnamed}
		\qquad
    \subfloat[$\cG_B$ bipartite]{
		\beginpgfgraphicnamed{black-bipartite}
		\begin{tikzpicture}
		[white/.style={circle,draw=black,fill=white,thin,inner sep=0pt,minimum size=1.2mm},
black1/.style={circle,draw=black,fill=black,thin,inner sep=0pt,minimum size=1.2mm},black2/.style={rectangle,draw=black,fill=black,thin,inner sep=0pt,minimum size=1.2mm}]
		\node[white] (w1) [label=above right:$\theta^*$] at (-1,-1) {};
		\node[white] (w2) [label=below left:$\theta^*$] at (1,1) {};
		\node[black1] (b1) [label=above left:$\theta$] at (1,-1) {};
		\node[black2] (b2) [label=below right:$\theta$] at (-1,1) {};
		\draw[decoration={markings,mark=at position 1 with {\arrow[scale=2]{>}}},
    postaction={decorate}, shorten >=0.4pt] (b1) -- node[midway,below] {$\exp(i\alpha)$} (w1);
		\draw[decoration={markings,mark=at position 1 with {\arrow[scale=2]{>}}},
    postaction={decorate}, shorten >=0.4pt] (b1) -- node[midway,right] {$\exp(i\beta)$} (w2);
		\draw[decoration={markings,mark=at position 1 with {\arrow[scale=2]{>}}},
    postaction={decorate}, shorten >=0.4pt] (b2) -- node[midway,left] {$-\exp(i\beta)$} (w1);
		\draw[decoration={markings,mark=at position 1 with {\arrow[scale=2]{>}}},
    postaction={decorate}, shorten >=0.4pt] (b2) -- node[midway,above] {$-\exp(i\alpha)$} (w2);
		\draw (-0.2,-1) arc (0:90:8mm);
		\draw (1,-0.2) arc (90:180:8mm);
		\draw (0.2,1) arc (180:270:8mm);
		\draw (-1,0.2) arc (270:360:8mm);
		\end{tikzpicture}
		\endpgfgraphicnamed}
   \caption[]{Induced labeling of a rhombic embedding}
   \label{fig:rhombic_bipartite}
\end{figure}
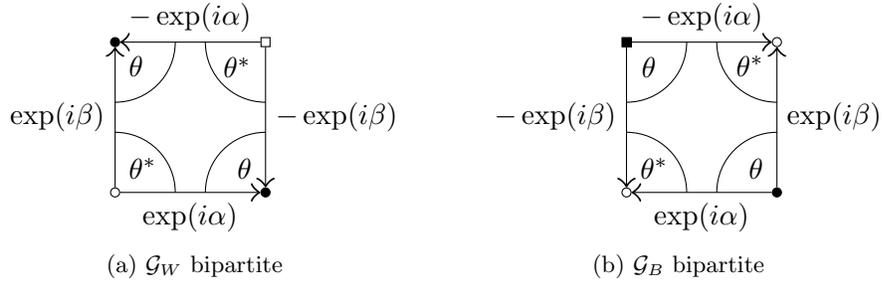

To each oriented edge $\vec{e}=(x,y)$ we associate the complex number \[\gamma(\vec{e}):=\frac{y-x}{\|y-x\|},\] where we now consider the vertices as points in $\mC$. If the white (or black) starting point of $e$ is of type 1, we define $\exp(i\alpha):=\gamma(\vec{e})$, otherwise $\exp(i\alpha):=-\gamma(\vec{e})$. By construction, the numbers $\alpha$ (that are unique up to multiples of $2\pi$) define labels on nonoriented edges of $\cG$ such that opposite edges of a quadrilateral receive the same label.

Now, the induced labels are given by $\alpha-\beta$. Clearly, $\alpha-\beta \equiv \theta \mod 2\pi$ (or $\alpha-\beta\equiv \theta^* \mod 2\pi$), where $\theta$ (or $\theta^*$) denotes the interior angle of the rhombus at the corresponding white (or black) vertex.
\end{proof}

\begin{remark}
If $\cG$ is simply-connected, $\cG_W$ (or $\cG_B$) is bipartite if and only if all inner black (or white) vertices are of even degree.
\end{remark}

For cellular decompositions of the disk, the converse of Proposition~\ref{prop:rhombic1} is true:

\begin{proposition}\label{prop:rhombic2}
Suppose that $\cG$ corresponds to a cellular decomposition of the disk and that $\cG_W$ or $\cG_B$ is bipartite. Then, any convex integrable labeling $\theta$ of $E(\cG_W)$ gives a (possibly ramified) rhombic embedding of $\cG$ into the complex plane.
\end{proposition}
\begin{proof}
Replacing $\theta$ by $-\theta$ if necessary, we can assume that $\exp(i\theta_e) \in S^1_+$ for all $e\in E(\cG_W)$. By assumption, we can choose labels $\alpha$ on the edges of $\cG$ that induce the given labeling on $E(\cG_W)$. Then, we can define $\gamma(\vec{e})$ through the labels $\alpha$ exactly as in the proof of Proposition~\ref{prop:rhombic1}. Choosing one starting point, the $\gamma(\vec{e})$ give us consecutively the positions of the vertices of $\cG$ in $\mC$. By construction, any quadrilateral is mapped to a rhombus, and by Proposition~\ref{prop:integrability} and the interpretation of $\theta$ in Proposition~\ref{prop:rhombic1}, the interior angles of the rhombi sum um to multiples of $2\pi$. Thus, we obtain a possibly ramified rhombic realization of $\cG$.
\end{proof}

The propositions above show that there is a large class of rhombic-embeddable quad-graphs for which convex integrable labelings exist. However, this is not true for all rhombic-embeddable quad-graphs.

\begin{proposition}\label{prop:rhombic-nonconvex}
There exist infinitely many rhombic-embeddable quad-graphs $\cG$ such that every induced labeling on $\cG_W$ is not convex.
\end{proposition}
\begin{proof}
It is easy to see that the graph shown in Figure~\ref{fig:rhombic-nonconvex} possesses a rhombic embedding according to Proposition~\ref{prop:rhombic} and that it can be easily extended to arbitrary large rhombic-embeddable quad-graphs. Let us show that every induced labeling on $E(\cG_W)$ is not convex.

\begin{figure}[htbp]
\begin{center}
\beginpgfgraphicnamed{rhombic-nonconvex}
\begin{tikzpicture}
[white/.style={circle,draw=black,fill=white,thin,inner sep=0pt,minimum size=1.2mm},
black/.style={circle,draw=black,fill=black,thin,inner sep=0pt,minimum size=1.2mm}]
\node[white] (w1)  at (0,-1) {};
\node[white] (w2)  at (1,0) {};
\node[white] (w3)  at (2,1) {};
\node[white] (w4)  at (0,3) {};
\node[white] (w5)  at (-3,4) {};
\node[white] (w6)  at (-3,2) {};
\node[white] (w7)  at (-4,1) {};
\node[white] (w8)  at (-2,-1) {};
\node[white] (w9)  at (-1,0) {};
\node[white] (w10)  at (0,1) {};
\node[white] (w11)  at (1,2) {};
\node[white] (w12)  at (-2,3) {};
\node[white] (w13)  at (-1,2) {};
\node[white] (w14)  at (-2,1) {};
\node[white] (w15)  at (-3,0) {};
\node[black] (b1)  at (1,-1) {};
\node[black] (b2)  at (2,0) {};
\node[black] (b3)  at (2,2) {};
\node[black] (b4)  at (-2,4) {};
\node[black] (b5)  at (-3,3) {};
\node[black] (b6)  at (-4,2) {};
\node[black] (b7)  at (-4,0) {};
\node[black] (b8)  at (-1,-1) {};
\node[black] (b9)  at (0,0) {};
\node[black] (b10)  at (1,1) {};
\node[black] (b11)  at (0,2) {};
\node[black] (b12)  at (-2,2) {};
\node[black] (b13)  at (-1,1) {};
\node[black] (b14)  at (-2,0) {};
\node[black] (b15)  at (-3,1) {};
\draw (w1) -- node[midway,below] {$\varepsilon$} (b1) -- (w2)-- node[near start,below] {$\delta$} (b2) -- (w3)-- node[midway,left] {$\alpha$} (b3) -- node[midway,above] {$\gamma$} (w4)-- (b4) -- node[near end,above] {$\beta$} (w5)-- node[midway,left] {$\delta$} (b5) -- (w6)-- node[near end,above] {$\varepsilon$} (b6) -- node[midway,left] {$\alpha$} (w7)-- (b7) -- node[midway,below] {$\beta$} (w8)-- (b8) -- node[midway,left] {$\varepsilon$} (w9)-- (b9) -- (w10)-- node[near start,below] {$\gamma$} (b10) -- node[midway,left] {$\alpha$} (w11)-- node[near end,below] {$\gamma$} (b11) -- (w12)-- (b12) -- (w13)-- node[midway,left] {$\alpha$} (b13) -- (w14)-- (b14) -- node[near end,above] {$\beta$} (w15)-- (b15) -- node[near start,above] {$\beta$} (w14)-- node[midway,left] {$\alpha$} (b12) -- node[near end,above] {$\beta$} (w6)-- node[midway,left] {$\alpha$} (b15) -- node[near end,above] {$\varepsilon$} (w7);
\draw (w9) -- node[midway,left] {} (b14)-- node[midway,left] {$\varepsilon$} (w8);
\draw (b8) -- node[midway,below] {$\gamma$} (w1);
\draw (b7) -- node[near start,above] {$\varepsilon$} (w15);
\draw (b3) -- node[near end,below] {$\delta$} (w11);
\draw (b11) -- node[midway,left] {$\delta$} (w4);
\draw (b5) -- node[near start,above] {$\beta$} (w12)-- node[midway,left] {$\delta$} (b4);
\draw (b9) -- node[near start,below] {$\gamma$} (w2)-- node[midway,left] {} (b10)-- node[near start,below] {$\delta$} (w3);
\draw (b1) -- node[near end,below] {$\gamma$} (w9)-- node[midway,left] {} (b13)-- node[midway,left] {} (w10)-- node[midway,left] {$\alpha$} (b11)-- node[midway,left] {} (w13);
\end{tikzpicture}
\endpgfgraphicnamed
\caption{rhombic-embeddable graph from Proposition~\ref{prop:rhombic-nonconvex}}
\label{fig:rhombic-nonconvex}
\end{center}
\end{figure}
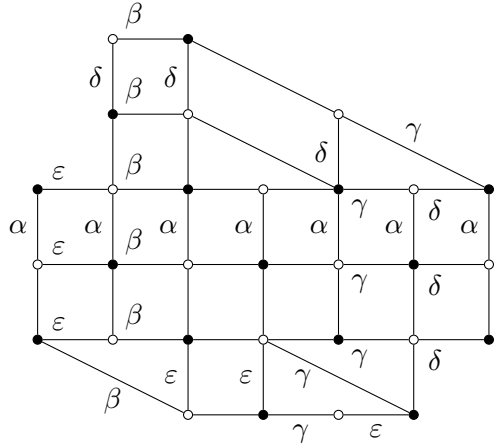

Assume the contrary and consider an appropriate labeling of $E(\cG)$. Again, we consider the labels modulo $2\pi$, such that we can assume that they are all in $[0,2\pi)$. Without loss of generality, we suppose that $\varepsilon=0$.

Due to convexity of the labeling, $\exp(i(\varepsilon-\mu)) \in S_+^1$ for all $\mu \in \{\alpha,\beta,\gamma\}$. Thus, $\alpha,\beta,\gamma \in (\pi,2 \pi)$. Moreover, $\exp(i(\alpha-\beta)),\exp(i(\gamma-\alpha)) \in S_+^1$, which gives $\gamma > \alpha > \beta$. Again by convexity of the labeling, the following restrictions for $\delta$ hold: $\exp(i(\alpha-\delta)),\exp(i(\delta-\beta)),\exp(i(\delta-\gamma)) \in S_+^1$. Using all the results we have obtained so far,
\begin{align*}
\delta & \in (\alpha-\pi,\alpha) \cap \{(\beta,2\pi) \cup [0,\beta-\pi)\} \cap \{(\gamma,2\pi) \cup [0,\gamma-\pi)\}\\
&=(\beta,\alpha) \cap \{(\gamma,2\pi) \cup [0,\gamma-\pi)\}\\
&=(\beta,\alpha) \cap (\gamma,2\pi)\\
&=\emptyset,
\end{align*}
contradiction.
\end{proof}

\begin{remark}
Note that the graph shown in Figure~\ref{fig:rhombic-nonconvex} is actually not that artificial. The proof why any labeling on its edges cannot induce a convex labeling mainly relies on the fact that there are a couple of triples of strips that are pairwise intersecting. The existence of such triples is reasonable to expect from a generic rhombic quad-graph. However, no three pairwise intersecting strips are all intersected by the same strip in the example above.
\end{remark}


\section{Integrable circle patterns and \texorpdfstring{$(\textnormal{Q}3)$}{(Q3)}}\label{sec:circle_patterns}

We conclude this paper by continuing our discussion of the analogies between the discrete Laplace-type equations corresponding to $(\textnormal{Q}3)$ and circle patterns as described by the first author and Springborn in \cite{BSp04}. Also, we comment on our and other notions of integrability.

\begin{figure}[htbp]
\begin{center}
  \beginpgfgraphicnamed{circle-intersection}
    \begin{tikzpicture}
    [white/.style={circle,draw=black,fill=white,thin,inner sep=0pt,minimum size=1.2mm},invisible/.style={circle,draw=white,fill=white,thin,inner sep=0pt,minimum size=1.2mm},black/.style={circle,draw=black,fill=black,thin,inner sep=0pt,minimum size=1.2mm}]
    \clip (-2.5,-3) rectangle (4.5,3);
    \node[invisible] (phi) [label=right:$\phi_e^x$]  at (-1.8,0.24) {};
    \node[invisible] (theta) [label=above:$\theta_e$]  at (0.1,-1.9) {};
    \node[white] (w1) [label=left:$x$]  at (-2,0) {};
    \node[white] (w2) [label=right:$y$] at (4,0) {};
    \node[black] (b1) at (0,-2) {};
    \node[black] (b2) at (0,2) {};
    \draw (b1) -- node[midway,below left] {$r_x$} (w1);
    \draw (b1) -- node[midway,below right] {$r_y$} (w2);
    \draw (b2) -- node[midway,above left] {$r_x$} (w1);
    \draw (b2) -- node[midway,above right] {$r_y$} (w2);
    \draw[dashed] (w1) -- node[midway,above] {$e$} (w2);
    \draw (-2,0) circle [radius=2.828cm];
    \draw (4,0) circle [radius=4.472cm];
    \draw (-1,0) arc (0:45:1cm);
    \clip (-2,0) -- (0,-2) -- (4,0) -- (-2,0);
    \draw (0,-2) circle [radius=0.7cm];
    \end{tikzpicture}
  \endpgfgraphicnamed
  \caption{Intersection of circles}
  \label{fig:circle-intersection}
\end{center}
\end{figure}
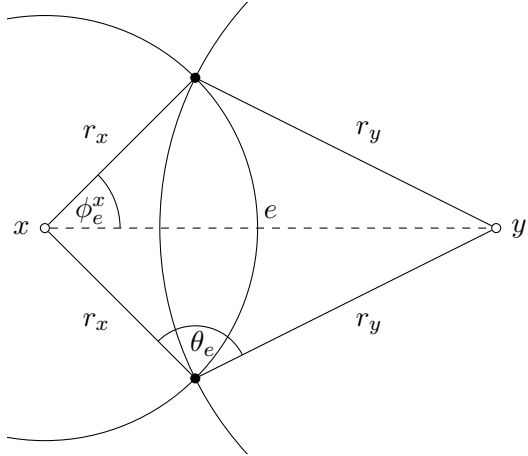

Figure~\ref{fig:circle-intersection} shows two circles centered at $x$ and $y$ with radii $r_x$ and $r_y$ that intersect at the angle $\theta_e$ (the circles being canonically oriented). Thus, the fact that the angles around an inner black vertex sum up to $2\pi$ (or a multiple of it, if we consider ramified patterns) corresponds to Equation~(\ref{eq:face}) in Proposition~\ref{prop:integrability}.

The first author and Springborn observed in \cite{BSp04} that elementary geometric calculations yield
\begin{equation}\label{eq:phi_eucl}
\exp(2i\phi_e^x)=\frac{r_x-r_y \exp(-i\theta_e)}{r_x-r_y \exp(i\theta_e)}
\end{equation}
in the Euclidean case, and
\begin{equation}
\exp(2i\phi_e^x)=\log\left(\frac{\tanh\left(\frac{r_x}{2}\right)-\tanh\left(\frac{r_y}{2}\right) \exp(-i\theta_e)}{\tanh\left(\frac{r_x}{2}\right)-\tanh\left(\frac{r_y}{2}\right) \exp(i\theta_e)}\right)-\log\left(\frac{1-\tanh\left(\frac{r_x}{2}\right)\tanh\left(\frac{r_y}{2}\right) \exp(-i\theta_e)}{1-\tanh\left(\frac{r_x}{2}\right)\tanh\left(\frac{r_y}{2}\right) \exp(i\theta_e)}\right) \label{eq:phi_hyp}
\end{equation}
in the hyperbolic case. Let us define
\begin{align}
 r_x=\exp(X),r_y=\exp(Y) &\text{ in the Euclidean case,} \label{eq:connection_eucl}\\
\tanh\left(\frac{r_x}{2}\right)=\exp(X),\tanh\left(\frac{r_y}{2}\right)=\exp(Y) &\text{ in the hyperbolic case,} \label{eq:connection_hyp}
\end{align}
$X,Y \in \mR$. This is possible because all radii are positive. Straightforward calculations show in both cases that
\begin{equation*}
2\phi_e^x=\varphi(X,Y;\theta_e),
\end{equation*}
where $\varphi$ is defined in (\ref{eq:phi_Q30}) for $(\textnormal{Q}3)_{\delta=0}$ and in (\ref{eq:phi_Q31}) for $(\textnormal{Q}3)_{\delta=1}$, corresponding to the Euclidean and hyperbolic case, respectively.

In Sections~\ref{sec:Q30} and~\ref{sec:Q31}, we have shown that in the case of $X\in \mR$ and $\theta_e \in (0,\pi)$ for all vertices $x \in V(\cG_W)$ and edges $e\in E(\cG_W)$ the terms $\Theta_x$ are positive, such that they can be seen as the cone angles at vertices $x\in V(\cG_W)$. Thus, the discrete Laplace-type Equations~(\ref{eq:Laplace}) of type $(\textnormal{Q}3)_{\delta=0}$ and $(\textnormal{Q}3)_{\delta=1}$ are identical to the variational formulation of Euclidean and hyperbolic circle patterns, respectively, given by the first author and Springborn in \cite{BSp04}.

Let us now come to integrability. By Proposition~\ref{prop:integrability}, the labeling $\theta$ of $E(\cG_W)$ is integrable if and only if 
\begin{equation}\label{eq:integrable_pattern1}
 \prod\limits_{e=(x,y_k)\in E(\cG_W)} \exp(i\theta_e) =1
\end{equation}
is true for all inner white vertices $x$, having in mind that Equation~(\ref{eq:face}) is satisfied since angles $\theta$ sum up to $2\pi$ around inner black vertices.

Assume now that the circle pattern can be partitioned into two sets of circles such that any two circles from one set have no points in common. Equivalently, $\cG_W$ is bipartite. If we denote the decomposition by $V(\cG_W)=V_1 \dot{\cup} V_2$, there is another way to relate the discrete Laplace-type equations corresponding to $(\textnormal{Q}3)$ with such circle patterns. For this, we use the other reality condition that differs to the previous ones by subtracting $\pi i$ to $X$ if $x \in V_2$. If we replace $X$ by $X-\pi i$ (or $Y$ by $Y-\pi i$) in (\ref{eq:connection_eucl}) or (\ref{eq:connection_hyp}), we have to replace $r_x$ by $-r_x$ (or $r_y$ by $-r_y$). Now,  $\exp(2i\phi_e^x)$ in Equations~(\ref{eq:phi_eucl}) and~(\ref{eq:phi_hyp}) then remains the same if we replace $\theta_e$ by $-\theta_e^*$.

In Sections~\ref{sec:Q30} and~\ref{sec:Q31}, we have explained why the variational formulation of the discrete Laplace-type equation for the modified reality condition is essentially the same as for the original one. Consequently, we obtain another interpretation of the discrete Laplace-type equations corresponding to $(\textnormal{Q}3)$, now with labels $-\theta^*$ instead of $\theta$, and circle patterns with intersection angles $\theta$.

This new labeling of $E(\cG_W)$ is then integrable if and only if
\begin{equation}\label{eq:integrable_pattern2}
 \prod\limits_{e=(x,y_k)\in E(\cG_W)} \exp(i\theta_e^*) =1
\end{equation}
is true for all inner black vertices $x$.

\begin{remark}
Equation~(\ref{eq:integrable_pattern1}) and (\ref{eq:integrable_pattern2}) are equivalent if and only if all vertices of $\cG_W$ have even degree. In the case that $\cG$ is simply-connected, we can equivalently state that the black graph $\cG_B$ is bipartite.
\end{remark}

If $\cG$ corresponds to a cellular decomposition of a disk, circle patterns satisfying Equation~(\ref{eq:integrable_pattern2}) for all inner black vertices $x$ are called \textit{integrable} in the paper \cite{BoMeSu05} of the first author, Mercat, and Suris. If and only if this equation is satisfied, the circle patterns admit isoradial realizations. Assuming that $\cG_W$ is bipartite (as it is necessary for the second reality condition), this statement is exactly the content of Propositions~\ref{prop:rhombic1} and~\ref{prop:rhombic2}.

Note that our definition of integrability of the corresponding labeling of $E(\cG_W)$ is denoted by \textit{integrability of the corresponding cross-ratio system} in \cite{BoMeSu05}, meaning that the system of Equations~(\ref{eq:Q}), where $Q$ is the cross-ratio equation $(\textnormal{Q}1)_{\delta=0}$, is three-dimensional consistent.

In their paper \cite{BoMeSu05}, the first author, Mercat, and Suris described how integrable circle patterns yield solutions of the corresponding cross-ratio system and under which conditions solutions of the cross-ratio system describe integrable circle patterns.


\section*{Acknowledgment}
\addcontentsline{toc}{section}{Acknowledgments}

We would like to thank V.V.~Bazhanov and Yu.~B.~Suris for numerous fruitful discussions.


\section*{Appendix: ABS list}
\addcontentsline{toc}{section}{Appendix: ABS list}

In the following, we will list the quad-equations of the ABS list that are of type Q and present the additive and multiplicative long leg functions $\varphi$ and $\Phi$, respectively.

In contrast to the original formulation of the equation $(\textnormal{Q}4)$ by the first author, Adler, and Suris in \cite{ABS03} that was based on the Weierstra{\ss} normalization of an elliptic curve, we use the Jacobian formulation that was first proposed by Hietarinta in \cite{Hi05}. The first author and Suris used the Jacobi normalization in the appendix of their paper \cite{BS10}. Our formulation follows their presentation, but differs slightly. We will indicate which (simple) transformations of variables and parameters relate our notation to the one in \cite{BS10}.

For parameters $\alpha,\beta$, let $\theta:=\alpha-\beta$. Note that the generalized action functionals of integrable quad-equations discussed in Section~\ref{sec:action} correspond to a general choice of parameters $\theta$.

{\bf List Q:}

\begin{itemize}
\item[] $(\textnormal{Q}1)_{\delta=0}$: $\quad
Q=\alpha(xu+yv)-\beta(xv+yu)-
  (\alpha-\beta)(xy+uv),$
  
\item[] $\qquad \qquad \:\quad i\varphi(x,y;\theta)=\frac{i\theta}{x-y}$.

\item[] Whereas the formulation of the quad-equation coincides with the one in \cite{BS10}, their long leg function \[i\varphi(x,y;\theta)=\frac{\theta}{x-y}\]
is related to ours by the transformation $\alpha \mapsto i\alpha$ for all parameters $\alpha$.

\item[]$(\textnormal{Q}1)_{\delta=1}$: $\quad
Q=\alpha(xu+yv)-\beta(xv+yu)-
  (\alpha-\beta)(xy+uv)-\alpha\beta(\alpha-\beta),$

\item[]
$\qquad \qquad \:\quad\Phi(x,y;\theta)=\frac{x-y+i\theta}{x-y-i\theta}.$

\item[] In \cite{BS10}, the term $+\alpha\beta(\alpha-\beta)$ instead of $-\alpha\beta(\alpha-\beta)$ appeared in the quad-equation, and their long leg function was given by \[\Phi(x,y;\theta)=\frac{x-y+\theta}{x-y-\theta}.\] The transformation $\alpha \mapsto i\alpha$ for all parameters $\alpha$ relates these two formulations.

\item[]$(\textnormal{Q}2)$:
 $\quad \ \quad Q=\alpha(xu+yv)-\beta(xv+yu)-(\alpha-\beta)(xy+uv)$

 $\quad \ \qquad\qquad  \quad -\alpha\beta(\alpha-\beta)(x+u+y+v)
           -\alpha\beta(\alpha-\beta)(\alpha^2-\alpha\beta+\beta^2),$

\item[] $\quad \: \qquad \qquad \: x=X^2$,

\item[]
$\quad \: \qquad \qquad \:\Phi(x,y;\theta)=\frac{(X-Y+i\theta)(X+Y+i\theta)}{(X-Y-i\theta)(X+Y-i\theta)}$.

\item[] As before, the transformation $\alpha \mapsto i\alpha$ for all parameters $\alpha$ relates the formulation in~\cite{BS10} to ours. There, $+\alpha\beta(\alpha-\beta)(x+u+y+v)$ instead of $-\alpha\beta(\alpha-\beta)(x+u+y+v)$ appeared in the quad-equation and their long leg function was \[\Phi(x,y;\theta)=\frac{(X-Y+\theta)(X+Y+\theta)}{(X-Y-\theta)(X+Y-\theta)}.\]

\item[]$(\textnormal{Q}3)_{\delta=0}$:
 $\quad Q=\sin(\alpha)(xu+yv)-\sin(\beta)(xv+yu)-\sin(\alpha-\beta)(xy+uv),$

\item[] $\qquad \qquad \:\quad x=\exp(X)$,

\item[]
$\qquad \qquad \:\quad\Phi(x,y;\theta)=\exp(-i\theta)\frac{\sinh\left(\frac{X-Y+i\theta}{2}\right)}
{\sinh\left(\frac{X-Y-i\theta}{2}\right)}$,

\item[]
$\qquad \qquad \:\quad\Phi'(x,y;\theta)=\frac{\sinh\left(\frac{X-Y+i\theta}{2}\right)}
{\sinh\left(\frac{X-Y-i\theta}{2}\right)}$.

\item[] Our quad-equation coincides with the one in \cite{BS10}, but they used the variable transformation $x=\exp(iX)$ instead. Now, the transformation $X \mapsto -iX$ for all new variables $X$ relates their long leg function given by \[\Phi(x,y;\theta)=\frac{\sin\left(\frac{X-Y+\theta}{2}\right)}{\sin\left(\frac{X-Y-\theta}{2}\right)}\] to ours if one does not take the factor $\exp(-i\theta)$ into account. However, multiplying the short leg functions by $\exp(-i\alpha)$ or $\exp(-i\beta)$, respectively, shows that these three-leg forms are equivalent.

\item[]$(\textnormal{Q}3)_{\delta=1}$:
 $\quad Q=\sin(\alpha)(xu+yv)-\sin(\beta)(xv+yu)
      -\sin(\alpha-\beta)(xy+uv)$

 $\quad \ \qquad\qquad  \quad+\sin(\alpha-\beta)\sin(\alpha)\sin(\beta),$

\item[] $\qquad \qquad \:\quad x=\cosh(X)$,

\item[]
$\qquad \qquad \:\quad\Phi(x,y;\theta)=\frac{\sinh\left(\frac{X-Y+i\theta}{2}\right)
\sinh\left(\frac{X+Y+i\theta}{2}\right)}{\sinh\left(\frac{X-Y-i\theta}{2}\right)
\sinh\left(\frac{X+Y-i\theta}{2}\right)}$.

\item[] Whereas our quad-equation is formulated in the same way as in \cite{BS10}, they used the variable transformation $x=\sin(X)$ instead. Their long leg function was given by \[\Phi(x,y;\theta)=\frac{\cos\left(\frac{X-Y+\theta}{2}\right)
\sin\left(\frac{X+Y+\theta}{2}\right)}{\cos\left(\frac{X-Y-\theta}{2}\right)
\sin\left(\frac{X+Y-\theta}{2}\right)};\] the transformation $X \mapsto \pi/2-iX$ for all new variables $X$ relates their long leg function to ours.

\item[]$(\textnormal{Q}4)$:
        $\quad \ \quad Q=\sn(\alpha)(xu+yv)-\sn(\beta)(xv+yu)
        -\sn(\alpha-\beta)(xy+uv)$

        $\quad \ \qquad\qquad  \quad+\sn(\alpha-\beta)\sn(\alpha)\sn(\beta)
        (1+\kappa^2xuyv),$

\item[] $\quad \: \qquad \qquad \: x=\sn(\pi/2-iX)$,

\item[] $\quad \: \qquad \qquad \:\Phi(x,y;\theta)=\frac
{\vartheta_1\left(\frac{X-Y+i\theta}{2i}\right)
\vartheta_1\left(\frac{X+Y+i\theta}{2i}\right)}
{\vartheta_1\left(\frac{X-Y-i\theta}{2i}\right)
\vartheta_1\left(\frac{X+Y-i\theta}{2i}\right)}$

$\quad \: \qquad \qquad \qquad \qquad \ \: =\frac
{\vartheta_1^*\left(\frac{X-Y+i\theta}{2i}\right)
\vartheta_4^*\left(\frac{X-Y+i\theta}{2i}\right)
\vartheta_1^*\left(\frac{X+Y+i\theta}{2i}\right)
\vartheta_4^*\left(\frac{X+Y+i\theta}{2i}\right)}
{\vartheta_1^*\left(\frac{X-Y-i\theta}{2i}\right)
\vartheta_4^*\left(\frac{X-Y-i\theta}{2i}\right)
\vartheta_1^*\left(\frac{X+Y-i\theta}{2i}\right)
\vartheta_4^*\left(\frac{X+Y-i\theta}{2i}\right)}.$

\item[] Here, $\vartheta_i^*$ is the Jacobi theta function whose half-period ratio is twice the ratio of $\vartheta_i$. The quad-equations are the same in both formulations, but $x=\sn(X)$ in~\cite{BS10}. The transformation $X \mapsto \pi/2-iX$ for all new variables $X$ relates the long leg function in \cite{BS10}, 
\[\Phi(x,y;\theta)=\frac
{\vartheta_2^*\left(\frac{X-Y+\theta}{2}\right)
\vartheta_3^*\left(\frac{X-Y+\theta}{2}\right)
\vartheta_1^*\left(\frac{X+Y+\theta}{2}\right)
\vartheta_4^*\left(\frac{X+Y+\theta}{2}\right)}
{\vartheta_2^*\left(\frac{X-Y-\theta}{2}\right)
\vartheta_3^*\left(\frac{X-Y-\theta}{2}\right)
\vartheta_1^*\left(\frac{X+Y-\theta}{2}\right)
\vartheta_4^*\left(\frac{X+Y-\theta}{2}\right)},\]
to ours.
\end{itemize}


\bibliographystyle{plain}
\bibliography{Discrete_integrable_eq_new}

\begin{thebibliography}{10}

\bibitem{ABS03}
V.E. Adler, A.I. Bobenko, and Yu.B. Suris.
\newblock Classification of integrable equations on quad-graphs. {T}he
  consistency approach.
\newblock {\em Commun. Math. Phys.}, 233(3):513--543, 2003.
\newblock Preprint: arXiv:nlin/0202024.

\bibitem{ABS09}
V.E. Adler, A.I. Bobenko, and Yu.B. Suris.
\newblock Discrete nonlinear hyperbolic equations. {C}lassification of
  integrable cases.
\newblock {\em Funct. Anal. Appl.}, 43(1):3--17, 2009.
\newblock Preprint: arXiv:0705.1663.

\bibitem{BaMSe07}
V.V. Bazhanov, V.V. Mangazeev, and S.M. Sergeev.
\newblock Faddeev-{V}olkov solution of the {Y}ang-{B}axter equation and
  discrete conformal symmetry.
\newblock {\em Nucl. Phys. B}, 784(3):234--258, 2007.
\newblock Preprint: arXiv:hep-th/0703041.

\bibitem{BaSe10}
V.V. Bazhanov and S.M. Sergeev.
\newblock A master solution of the quantum {Y}ang-{B}axter equation and
  classical discrete integrable equations.
\newblock Preprint: arXiv:1006.0651, 2010.

\bibitem{BaSe11}
V.V. Bazhanov and S.M. Sergeev.
\newblock Elliptic gamma-function and multi-spin solutions of the
  {Y}ang-{B}axter equation.
\newblock Preprint: arXiv: 1106.5874, 2011.

\bibitem{BoMeSu05}
A.I. Bobenko, C.~Mercat, and Yu.B. Suris.
\newblock Linear and nonlinear theories of discrete analytic functions.
  {I}ntegrable structure and isomonodromic {G}reen's function.
\newblock {\em J. Reine Angew. Math.}, 583:117--161, 2005.
\newblock Preprint: arXiv:math/0402097.

\bibitem{BSp04}
A.I. Bobenko and B.A. Springborn.
\newblock Variational principles for circle patterns and {K}oebe's theorem.
\newblock {\em Trans. Amer. Math. Soc.}, 356(2):659--689, 2004.
\newblock Preprint: arXiv:math/0203250.

\bibitem{BS02}
A.I. Bobenko and Yu.B. Suris.
\newblock Integrable systems on quad-graphs.
\newblock {\em Intern. Math. Research Notices}, 2002(11):573--611, 2002.
\newblock Preprint: arXiv:nlin/0110004.

\bibitem{BS08}
A.I. Bobenko and Yu.B. Suris.
\newblock {\em Discrete {D}ifferential geometry: {I}ntegrable {S}tructure},
  volume~98 of {\em Grad. Stud. Math.}
\newblock AMS, Providence, 2008.
\newblock Preprint: arXiv:math/0504358.

\bibitem{BS10}
A.I. Bobenko and Yu.B. Suris.
\newblock On the {L}agrangian structure of integrable quad-equations.
\newblock {\em Lett. Math. Phys.}, 92(3):17--31, 2010.
\newblock Preprint: arXiv:0912.2464.

\bibitem{BoSu14}
A.I. Bobenko and Yu.B. Suris.
\newblock Discrete pluriharmonic functions as solutions of linear
  pluri-{L}agrangian systems.
\newblock Preprint: arXiv: 1403.2876, 2014.

\bibitem{BPSu14}
R.~Boll, M.~Petrera, and Yu.B. Suris.
\newblock What is integrability of discrete variational systems?
\newblock {\em Proc. Royal Soc. A}, 470(2162), 2014.
\newblock 20130550 (15 pp). {P}reprint: arXiv:1307.0523.

\bibitem{Hi05}
J.~Hietarinta.
\newblock Searching for {CAC}-maps.
\newblock {\em J. Nonlin. Math. Phys.}, 12(Suppl. 2):223--230, 2005.

\bibitem{H00}
A.~Hurwitz.
\newblock {\em Vorlesungen \"uber allgemeine {F}unktionentheorie und
  elliptische {F}unktionen}.
\newblock Springer-Verlag, Berlin, 2000.
\newblock 5. Auflage.

\bibitem{KSch05}
R.~Kenyon and J.-M. Schlenker.
\newblock Rhombic embeddings of planar quad-graphs.
\newblock {\em Trans. Amer. Math. Soc.}, 357(9):3443--3458, 2005.
\newblock Preprint: arXiv:math-ph/0305057.

\bibitem{LN09}
S.~Lobb and F.W. Nijhoff.
\newblock Lagrangian multiforms and multidimensional consistency.
\newblock {\em J. Phys. A: Math. Theor.}, 42(45), 2009.
\newblock 454013 (18pp). {P}reprint: arXiv:0903.4086.

\bibitem{MaW97}
J.E. Marsden and J.M. Wendlandt.
\newblock Mechanical integrators derived from a discrete variational principle.
\newblock {\em Physica D}, 106(3-4):223--246, 1997.

\bibitem{N02}
F.W. Nijhoff.
\newblock Lax pair for the {A}dler (lattice {K}richever-{N}ovikov) system.
\newblock {\em Phys. Lett. A}, 297(1-2):49--58, 2002.
\newblock Preprint: arXiv:nlin/0110027.

\bibitem{Su13}
Yu.B. Suris.
\newblock Variational formulation of commuting {H}amiltonian flows: multi-time
  {L}agrangian 1-forms.
\newblock {\em J. Geom. Mechanics}, 5(3):365--379, 2013.
\newblock Preprint: arXiv:1212.3314.

\end{thebibliography}

\end{document}